\documentclass[letterpaper,10 pt, conference]{ieeeconf}

\IEEEoverridecommandlockouts                              
\usepackage{graphics} 
\usepackage{epsfig} 

\usepackage{amsmath,amssymb,amsfonts,amsthm} 
\usepackage{mathtools,bbm}
\usepackage{tikz, epic,eepic}
\usepackage{caption}
\usepackage{subcaption}
\usetikzlibrary{shapes,arrows}
\usetikzlibrary{automata}
\usepackage{bm}
\usepackage{epsfig, amsbsy}
\usepackage{mathtools}
\usepackage{graphicx}
\usepackage{caption}
\usepackage{subcaption}
\usepackage{algorithm}
\usepackage[noend]{algpseudocode}

\let\labelindent\relax
\usepackage{enumitem}
\theoremstyle{remark}

\newtheorem{proposition}{Proposition}

{\itshape}{\rmfamily}

\newcommand{\argmax}[1]{\underset{#1}
{\operatorname{arg}\,\operatorname{max}}\;}
\def\mcal{\mathcal}
\def\mbb{\mathbb}

\title{\LARGE \bf
Optimal control policies for evolutionary dynamics with environmental feedback
}

\author{Keith Paarporn$^{1}$, Ceyhun Eksin$^{2}$, Joshua S. Weitz$^{3,4,1}$, Yorai Wardi$^{1}$ 
\thanks{$^{1}$School of Electrical and Computer Engineering, Georgia Institute of Technology, Atlanta, GA 30332
        {\tt\small kpaarporn@gatech.edu, ywardi@ece.gatech.edu}}%
\thanks{$^{2}$Industrial \& Systems Engineering Department, Texas A\&M University, College Station, TX 77843 {\tt\small eksinc@exchange.tamu.edu}}
\thanks{$^{3}$School of Biological Sciences, Georgia Institute of Technology, Atlanta, GA}
\thanks{$^{4}$School of Physics, Georgia Institute of Technology, Atlanta, GA  {\tt\small jsweitz@gatech.edu}}%
}

\begin{document}

\maketitle
\thispagestyle{empty}
\pagestyle{empty}

\begin{abstract}
We study a dynamical model of a population of cooperators and defectors whose actions have long-term consequences on environmental ``commons'' - what we term the ``resource''. Cooperators contribute to restoring the resource whereas defectors degrade it. The population dynamics evolve according to a replicator equation coupled with an environmental state. Our goal is to identify methods of influencing the population with the objective to maximize accumulation of the resource. In particular, we consider strategies that modify individual-level incentives. We then extend the model to incorporate a public opinion state that imperfectly tracks the true environmental state, and study strategies that influence opinion. We formulate optimal control problems and solve them using numerical techniques to characterize locally optimal control policies for three problem formulations: 1) control of incentives, and control of opinions through 2) propaganda-like strategies and 3) awareness campaigns. We show numerically that the resulting controllers in all formulations achieve the objective, albeit with an unintended consequence. The resulting dynamics include cycles between low and high resource states - a dynamical regime termed an ``oscillating tragedy of the commons''. This outcome may have desirable average properties, but includes risks to resource depletion. Our findings suggest the need for new approaches to controlling coupled population-environment dynamics.

\end{abstract}
\section{Introduction}\label{sec:intro}

A tragedy of the commons occurs when individuals in a population are driven by their own selfish interests, resulting in the depletion of a common resource on which they all depend. The interactions that drive such tragedies are modeled in classical game theory as a prisoner's dilemma \cite{Dawes_1975,Ostrom_1990,Weitz_2016}. The rational choice for an individual is to defect, regardless of what others are doing. However, classical models do not account for the consequences of action - individual actions affect the environment. Consequently, the state of the environment may shape individual incentives for future action. Dynamical models of these coevolutionary features have been developed  to understand general conditions under which tragedies will occur or be averted \cite{Weitz_2016,Brown_2007,Roopnarine_2013}.  Similarly, the study of common-pool resource games suggest that rational play among larger populations leads to resource collapse with higher probability \cite{Hota_2016,Rapoport_1992}.

In his landmark paper \cite{Hardin_1968}, Hardin argues that such tragedies are inevitable given a growing human population, unless preventative measures are taken. To address the problem of preventing tragedies, there has been speculation about what intervention strategies will be effective. Interventions from centralized government entities are called for, through implementing and enforcing new policies restricting overconsumption \cite{Ostrom_1990,Aronson_2007}. For example, imposing taxes on resource usage may provide a financial deterrent to overuse \cite{Hota_2016_CDC}. Passing regulatory laws on fishers gives fish populations a chance to recover \cite{Hutchings_2004}. Hence, such direct intervention policies provide the incentives necessary to instigate conservation behaviors \cite{Penn_2003}.

Information also plays an important role. Individuals may not take pro-environmental actions if they are not informed about why such actions are necessary \cite{VanVugt_2009}.   Environmental awareness and education can lead to behavior changes when individuals realize that environmental degradation has adverse effects on their own community or household. For example, information from household metering about the severity of water scarcity drove efforts to conserve water \cite{VanVugt_1999}. However, statistics and facts may be ineffective to instigate behavior changes if such issues are politicized \cite{Lupia_2013}. In these situations, public opinion is susceptible to propaganda from news outlets and social media.  Environmental information is necessary to affect behavior change, but may not be sufficient \cite{Stern_2000,Penn_2003,VanVugt_2009}. The efficacy of these proposed solutions are rarely tested using dynamical models that couple actions and environmental changes \cite{Ostrom_1990}. 

A taxation mechanism on resource investment was studied in the setting of a common-pool resource game where under certain conditions, higher tax rates can lead to lower probability of resource collapse \cite{Hota_2016_CDC}. However, asymptotic outcomes are not considered in this static one-shot game. In a recent work \cite{Manzoor_2014}, an infinite horizon optimal control framework was applied to a dynamical model\footnote{Those dynamics can be reduced to a linear system by an appropriate transformation.  This differs fundamentally from the dynamical system considered in this paper, which is highly nonlinear and cannot be transformed into a linear system.} to identify conditions under which an optimal prescribed consumption rate ensures resource sustainability. However, the consumption rate is not directly manipulated by taxing, pricing, or other social control policies.

In contrast, we consider in this paper such direct control policies. We formulate optimal control problems that study the role of incentive and information-based intervention policies with the objective of maximally conserving the environmental state over a finite time horizon. We apply these control formulations to the model of ref. \cite{Weitz_2016},  due to its general framework. It models a population of myopic individuals whose actions affect and are affected by the environment. This framework differs from that of differential games \cite{Basar_1999}, where individuals select strategies to maximize long-term payoffs given action-dependent dynamic environments.  We formulate an incentive control problem by allowing an external entity to influence the population's incentive to cooperate together. To implement information-based control policies, we introduce a dynamic public opinion that imperfectly tracks the true environmental state. We present two formulations in which the control directly affects public opinion: propaganda strategies that perturb public opinion, and awareness-raising strategies where learning of the true environmental state is encouraged. In all three formulations, we compute optimal controls by numerical means (by ``optimal" in this paper, we mean locally optimal since the problems we formulate are nonconvex).

The main contributions and findings of this paper are 1) the formulation of optimal control problems to address the tragedy of the commons through direct policy interventions and 2) the solutions of these problems, obtained by numerical techniques, result in highly oscillatory behavior. In particular, we show through simulations that the objectives of the formulated problems are achieved, at the expense of inducing highly variant dynamics characterized by oscillatory cycles between low and high resource states.

The paper is organized as follows.  Section \ref{sec:model} presents the feedback-evolving game model of \cite{Weitz_2016}. Section \ref{sec:control_incentive} formulates the incentive optimal control problem, and presents numerical results from applying a suitable control algorithm \cite{Hale_2016}. We prove in this formulation that an optimal controller is necessarily bang-bang. In Section \ref{sec:control_information}, we introduce the public opinion dynamics, formulate the propaganda and awareness-raising control problems, and present numerical results. Concluding remarks and discussion points are given in Section \ref{sec:conclusion}.

\section{Model}\label{sec:model}

\subsection{Feedback-evolving games}

Here, we review the model of \cite{Weitz_2016}, which incorporates environmental feedback into replicator dynamics of a $2\times 2$ game, where the strategies are cooperate ($\mcal{C}$) and defect ($\mcal{D}$). This model is intended to provide a general framework in which to portray the dynamics of tragedy of the commons scenarios.  It incorporates an environment state $n(t)\in[0,1]$ where $n=0$ ($n=1$) means the environment is completely depleted (replenished). We will use the terms environment and common resources interchangeably to refer to $n$. The game payoffs are determined by $n$ as follows. 
\begin{equation}\label{eq:An}
	A_n =  \left[\begin{array}{cc} R_n  & S_n \\ T_n & P_n \end{array}\right]  \equiv n\left[\begin{array}{cc} R_1 & S_1 \\ T_1 & P_1 \end{array}\right] + (1-n)\left[\begin{array}{cc} R_0  & S_0 \\ T_0 & P_0 \end{array}\right].
\end{equation}
\begin{figure}
	\centering
	\includegraphics[scale=.55]{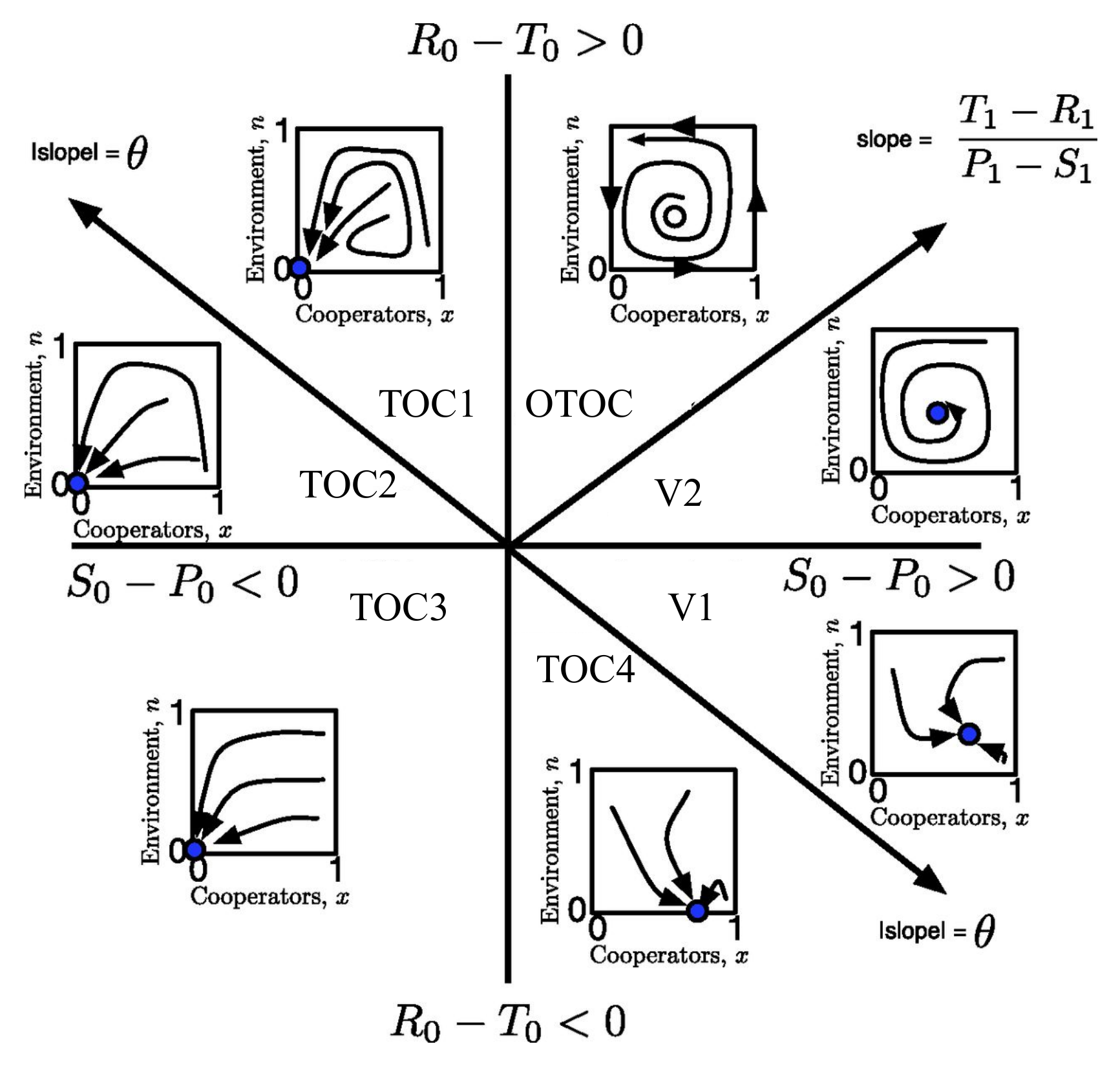}
	\caption{(Adapted from \cite{Weitz_2016}) Summary of all possible dynamical outcomes given choice of payoffs in deplete state. The regions are determined by the relative payoffs $S_0-P_0$ (x-axis) and $R_0 - T_0$ (y-axis). The phase portraits are illustrated in each region, where blue dots indicate stable fixed points of the dynamics. The seven regions include outcomes where a tragedy of the commons (TOC) occurs, and where a TOC is averted. We assign labels to each region, which includes four TOC outcomes, two averted outcomes (V1 and V2), and one oscillating TOC (OTOC). Here, the white dot indicates an unstable fixed point.}
	\label{fig:region_labels}
\end{figure}
When $n=1$, agents play a game determined by the payoff matrix $A_1$, given by the first matrix of the right-hand side above. Similarly, when $n=0$, the game is determined by the second matrix above, $A_0$. In the game with payoffs $A_1$, we impose that defection is the dominant strategy, that is, $R_1<T_1$ and $S_1 < P_1$. Thus, players will always prefer to defect when resources are abundant. The only pure Nash equilibrium in this game is mutual defection, where players obtain a payoff $P_1$. The structure of the game in the depleted state, given by the payoff matrix $A_0$, is a free parameter to allow different asymptotic outcomes of the system.  The frequency-dependent fitnesses for cooperators and defectors are therefore
\begin{equation}\label{eq:fitnesses}
	\begin{aligned}
		f_{\mcal{C}}(x,n) &= R_nx + S_n(1-x) \ \text{(cooperator fitness)}\\
		f_{\mcal{D}}(x,n) &= T_nx + P_n(1-x)  \ \text{(defector fitness)}
	\end{aligned}
\end{equation}
where $x\in[0,1]$ is the fraction (frequency) of cooperators in the population, and $1-x$ the fraction of defectors. The game-environment coupled dynamics obey the following differential equations.
\begin{equation}\label{eq:xn_ODE}
	\begin{aligned}
		&\dot{x} = x(1-x)g(x,n)  \\
		&\dot{n} = n(1-n)(\theta x - (1-x)) \\
		&x_0,n_0 \in [0,1]
	\end{aligned}
\end{equation}
where $g(x,n) \equiv f_{\mcal{C}}(x,n) - f_{\mcal{D}}(x,n)$ is the difference in fitness between cooperators and defectors. The $n(1-n)$ term indicates a logistic growth of the environmental state, and serves to constrain the dynamics to $n(t) \in [0,1]$ $\forall t \geq 0$. The growth or decline of the environment depends on the fraction $x$ of cooperators in the population, who enhance $n$ at a rate $\theta>0$ while defectors degrade $n$ at a rate $-1$.
We denote the state vector $\bm{y}(t) \equiv [x(t),n(t)]^\top$ and the system mapping of \eqref{eq:xn_ODE} as $F : [0,1]^2 \rightarrow \mbb{R}^2$.

 There are four ``corner" fixed points, (0,0), (1,0), (0,1), and (1,1). When $x = 0$, the trajectory is confined to the left edge of the state space, and converges to the equilibrium (0,0). When $x=1$, it is on the right edge and converges to (1,1). When $n=0$, the dynamics follow a replicator dynamic corresponding to the base game $A_0$, and when $n=1$, the dynamic converges to (0,1) since this corresponds to replicator dynamics of the PD game. However, we focus our attention on system dynamics in the interior of the state space $(0,1)^2$, which is forward invariant.

\subsection{Summary of dynamics in feedback-evolving games}

The behavior of the system \eqref{eq:xn_ODE} relies on the choice of the payoff parameters $R_0,S_0,T_0$, and $P_0$ of the game $A_0$. There are seven possible dynamical regimes, and they are summarized and named in Figure \ref{fig:region_labels}. The outcomes that are possible include a tragedy of the commons (TOC1 - TOC4), aversion of TOC (V1 and V2), and an ``oscillating" TOC (OTOC). In V2, trajectories asymptotically approach an interior fixed point. In OTOC, trajectories approach an asymptotically stable heteroclinic cycle, defined by the counter-clockwise orientation of the corners and the edges connecting them (see SI of \cite{Weitz_2016} for details). This dynamical outcome is termed an ``oscillating tragedy of the commons''  because it is characterized by cycles between replete and deplete environmental states.

\section{Incentive control policies}\label{sec:control_incentive}

\subsection{Optimal control formulation}

We consider here strategic policies that influence individuals' incentives to cooperate together with the goal of conserving public resources over time. The control variable $u(t)$ is applied to the payoff matrix \eqref{eq:An} as follows.
\begin{equation}\label{eq:Anu}
	A_n(u(t)) =  n\left[\begin{array}{cc} R_1 & S_1 \\ T_1 & P_1 \end{array}\right] + (1-n)\left[\begin{array}{cc} R_0+u(t)  & S_0 \\ T_0 & P_0 \end{array}\right].
\end{equation}
In this formulation, we will constrain $u(t) \in [-u_m,u_m]$ for all $t$, where $u_m > 0$ is a positive constant. We formulate the following optimal control problem in Bolza form with no terminal cost.
\begin{equation}\label{eq:control_incentive}
	\begin{aligned}
		&\max_{u} J=  \int_0^{T_f} n^{2}(t) dt \\
		&\text{subject to } \begin{cases} \dot x = x(1-x)g(x,n) + x^2(1-x)(1-n)u \\
		\dot n =n(1-n)(-1+(1+\theta)x) \\
		x_0,n_0 \in (0,1) \\
		u(t) \in [-u_m,u_m] \ \forall t \in [0,T_f] \end{cases}
	\end{aligned}
\end{equation}
Recall that $x$ and $n$ are the state variables with $u$ as the control. The term $x^2(1-x)(1-n)$ appears after re-deriving the replicator equation with the payoff \eqref{eq:Anu} in the same manner as \eqref{eq:fitnesses}  with the payoff matrix $A_n(u)$. The Hamiltonian of this formulation is
\begin{equation}\label{eq:Hamiltonian_incentive}
		\begin{aligned}
			H(\bm{y},\bm{\lambda},u) = &\lambda_x x(1-x)( g(x,n) + x(1-n)u) \\
			&+ \lambda_n n(1-n)(-1 + (1+\theta)) + n^2
		\end{aligned}
\end{equation}
The first-order optimality conditions required by Pontryagin's Maximum Principle (PMP) are given by the co-state dynamical equations
\begin{equation}\label{eq:costate_incentive}
	\begin{aligned}
		&\dot{\lambda}_x = -\frac{\partial H}{\partial x}(\bm{y},\bm{\lambda},u) \\
		&\dot{\lambda}_n = -\frac{\partial H}{\partial n}(\bm{y},\bm{\lambda},u) \\
		&\lambda_x(T_f) = \lambda_n(T_f) = 0
	\end{aligned}
\end{equation}
and the pointwise maximizer of the Hamiltonian
\begin{equation}\label{eq:ustar_incentive}
		u^*(t) = \begin{cases}
				u_m \ &\text{if } \varphi(t) > 0 \\
				? &\text{if } \varphi(t) = 0 \\
				-u_m &\text{if } \varphi(t) < 0
				\end{cases}
\end{equation}
where $\varphi(t) \equiv  x^2(1-x)(1-n) \lambda_x(t)$ is the switching function. In the case when $\varphi(t) = 0$, the Hamiltonian is independent of $u$, and hence $u^*$ can take an arbitrary value. Note that $x^2(1-x)(1-n) > 0$ for all $t$ because of invariance of the interior. As long as $\varphi(t) = 0$ does not occur on an open interval in the time horizon $[0,T_f]$, $u^*(t)$ is a bang-bang controller (no singular arcs). Hence, $u^*(t)$ will only take two values -  the minimum and maximum points in the constraint set $[-u_m,u_m]$. We prove that this is indeed true, using the Lie bracket to rule out the existence of any singular arcs (Ch. 4.4 of \cite{Liberzon}). At the isolated switching times, we may assume $u^*(t)$ takes one of the two values $\{-u_m,u_m\}$ to enforce one-sided continuity.
\begin{figure*}[t!]
	\centering
	\begin{subfigure}[t]{\columnwidth}
		\includegraphics[scale=.3]{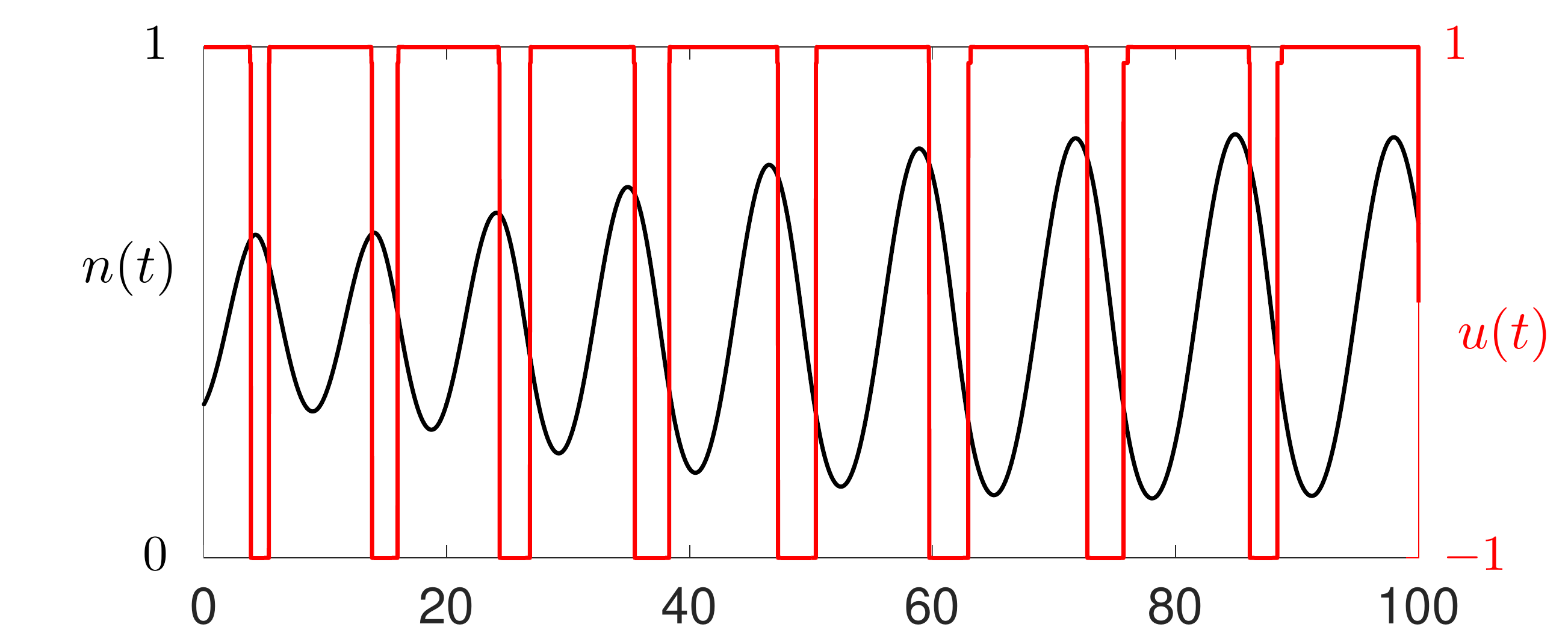}
		\includegraphics[scale=.25]{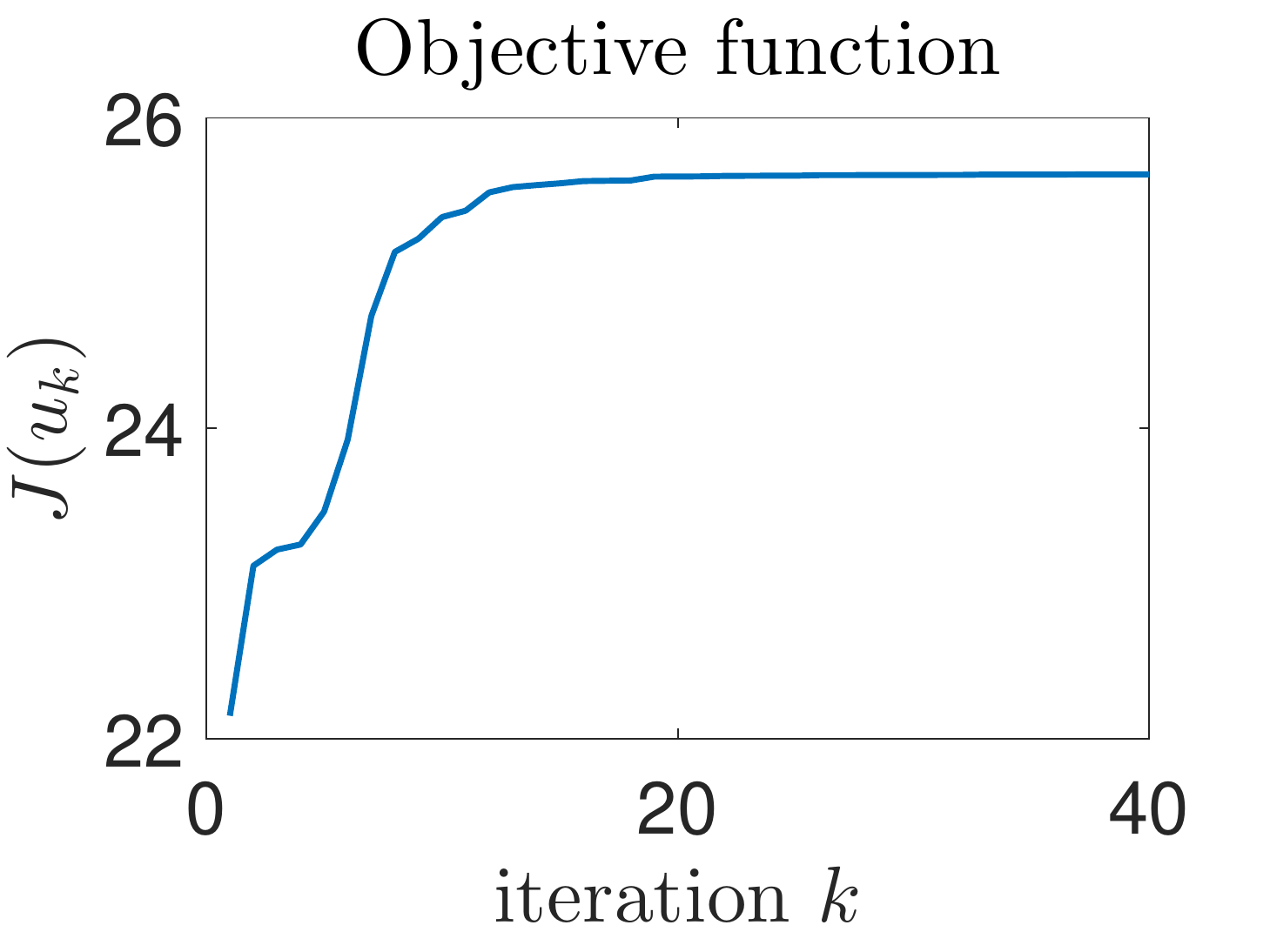}
		\includegraphics[scale=.25]{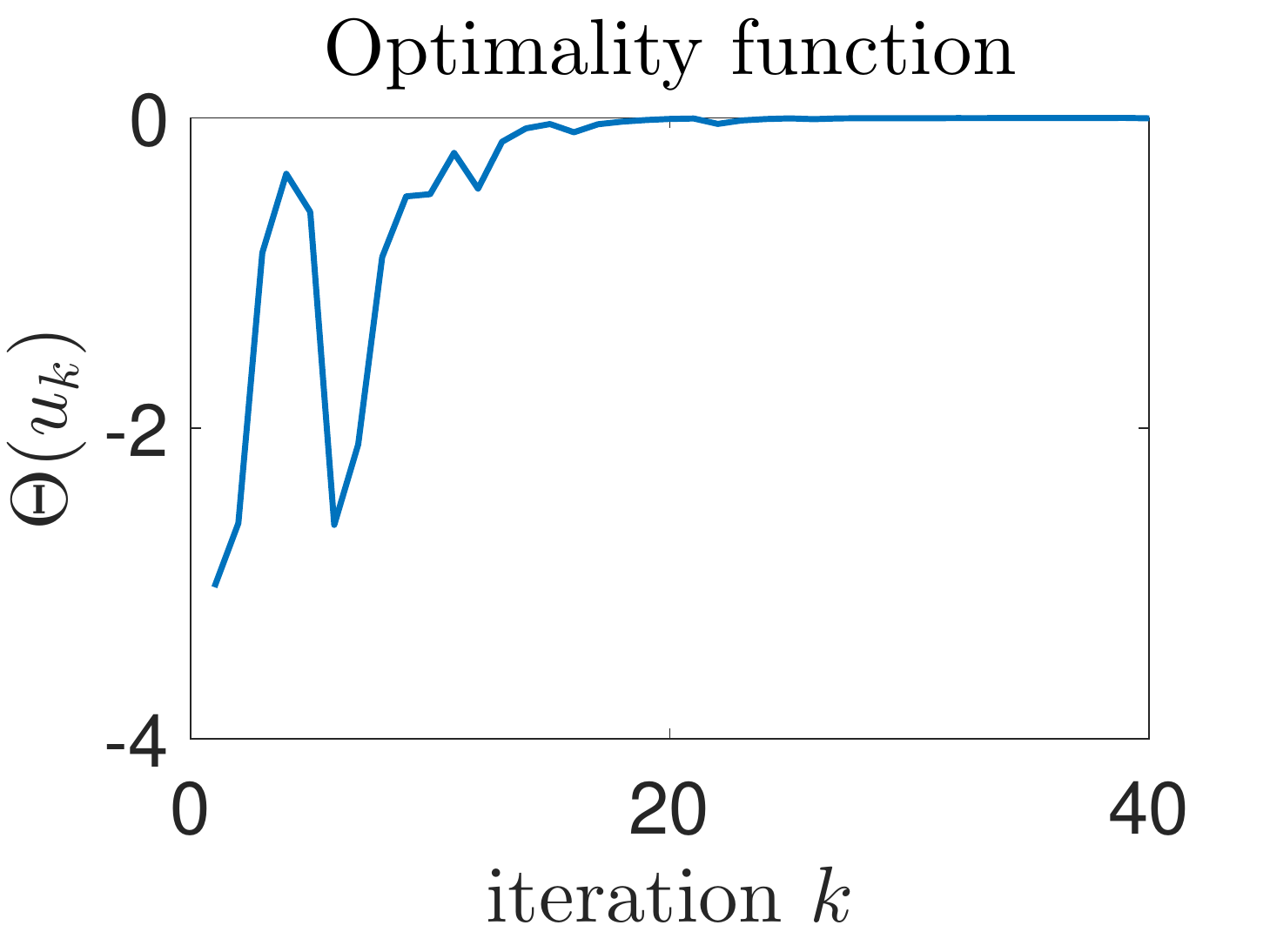}
		\centering
		\caption{Initial guess $u_0(t) = 0$.}
		\label{fig:u0_incentive}
	\end{subfigure}  
	\begin{subfigure}[t]{\columnwidth}
		\includegraphics[scale=.3]{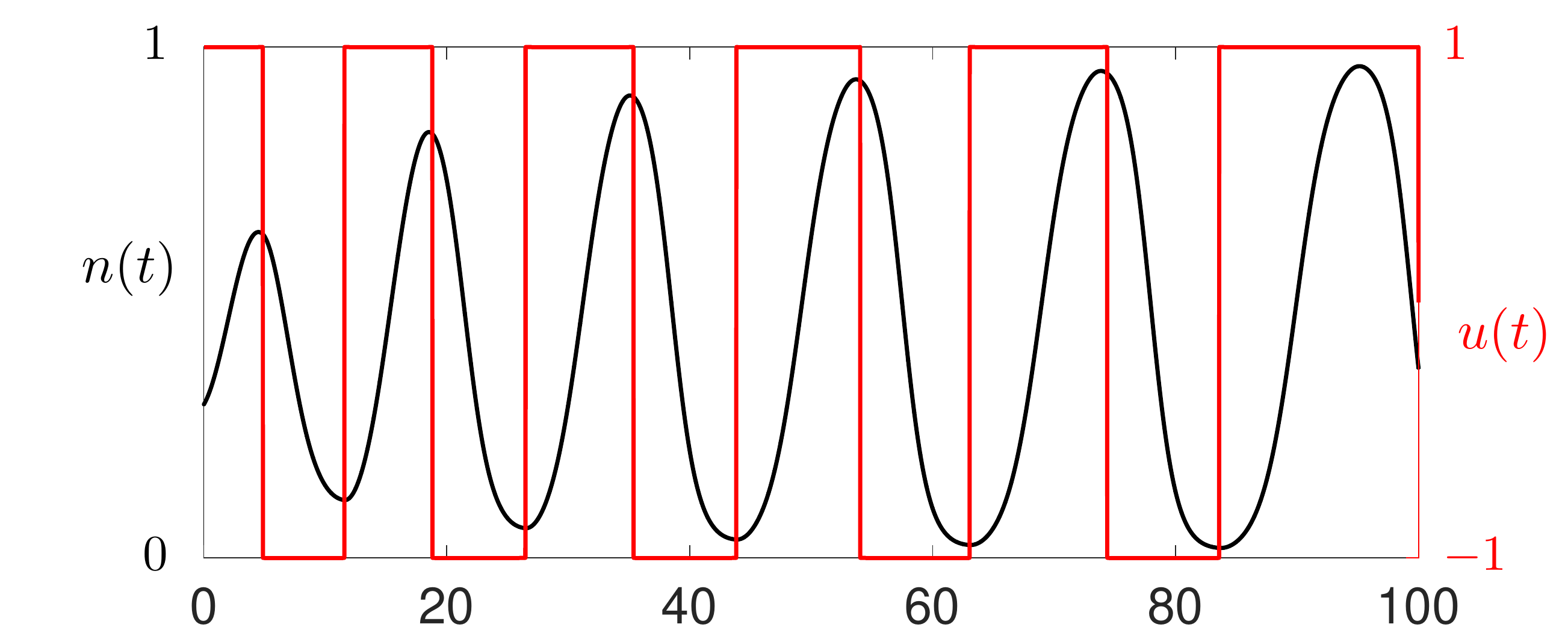}
		\includegraphics[scale=.25]{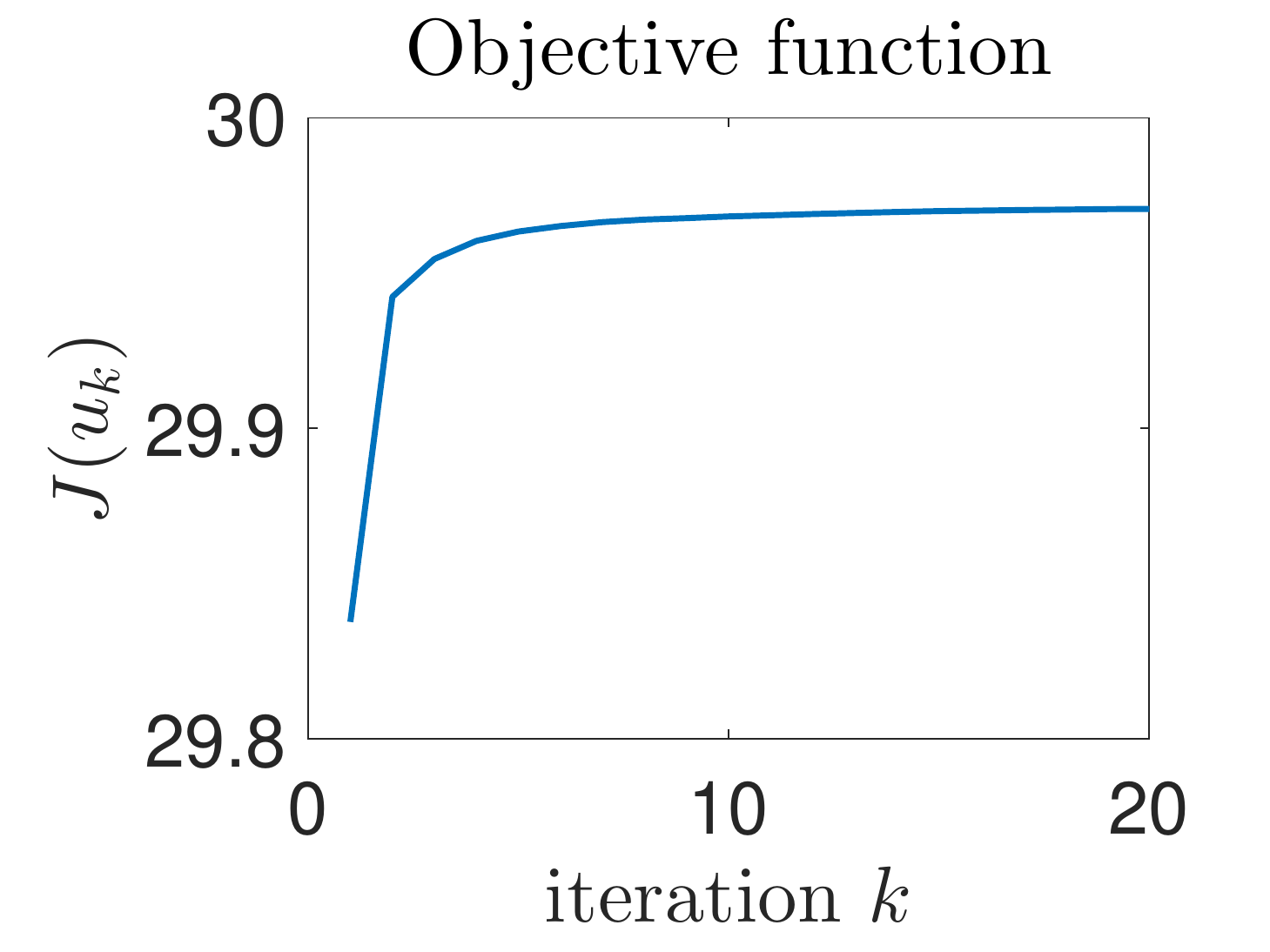}
	\includegraphics[scale=.25]{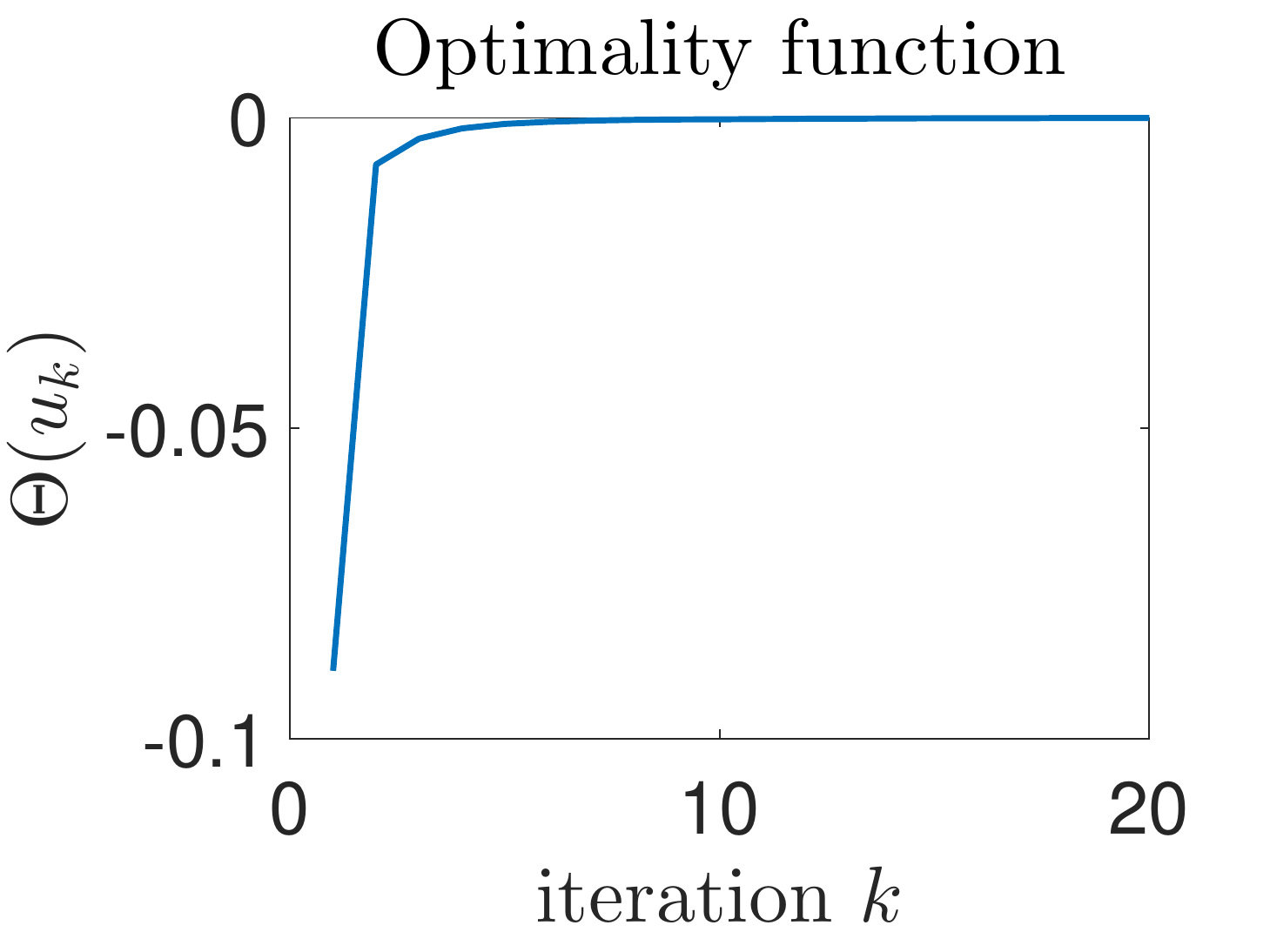}
		\centering
		\caption{State-dependent initial guess $u_0(t) = \text{sgn}(x(t)-x_c)$. }
		\label{fig:ubb_incentive}
	\end{subfigure}  
	\caption{Simulation results from applying Algorithm \ref{alg:algorithm} with $[R_0,S_0,T_0,P_0] = [4.5,4,3,3]$ and $u_m = 1$ to incentive control problem \eqref{eq:control_incentive}.  In left panels (a), we applied 40 iterations with $u_0(t) = 0$ (runtime 485 s). (Top) Environment dynamics $n(t)$ (black) overlayed with the resulting control $u_{40}$ (red). (Bottom Left) Objective scores $J(u_k) = \int_0^{T_f} n^2(t) dt$ vs iteration number $k$, where $J(u_{40}) =25.6359$. (Bottom Right) The optimality function $\Theta(u_k)$ (eq \eqref{eq:Theta} in Appendix) vs iteration number $k$, where $\Theta(u_{40}) \approx -0.0033$. In right panels (b), we set $u_0(t) =  \text{sgn}(x-x_c)$, and run 20 iterations (runtime 103 s). We obtain $J(u_{20}) =29.9707$ and $\Theta(u_{20}) =  -1.95 \times 10^{-5}$.  }
	\label{fig:V2_incentive_u0}
\end{figure*}
\begin{proposition}
	An optimal controller $u^*$ given by \eqref{eq:ustar_incentive} is non-singular. That is, it switches between the two values $\{-u_m,u_m\}$ at isolated points in the horizon interval $[0,T_f]$.
\end{proposition}
\begin{proof}
	The switching function can be written as the inner product
	\begin{equation}
		\varphi(t) = \langle \bm{\lambda}, G \rangle
	\end{equation}
	where $G = [x^2(1-x)(1-n),0]^\top$ is the control-affine vector field. The time derivative is given by
	\begin{equation}
		\dot{\varphi}(t) = \langle \bm{\lambda}, [F,G] \rangle
	\end{equation}
	where $F$ is the state vector field and 
	\begin{equation}
		[F,G] =  \frac{\partial G}{\partial \bm{y}}F - \frac{\partial F}{\partial \bm{y}}G
	\end{equation}
	is the Lie bracket of the vector fields $F$ and $G$. For non-singularity to hold for $u^*$, $\dot{\varphi}$ cannot be zero when $\varphi(t) = 0$. This is equivalent to proving the vector fields $G$ and $[F,G]$ are linearly independent. After some calculation, this amounts to checking independence for the vectors
	\begin{equation}
	\begin{bmatrix} 1 \\ 0 \end{bmatrix},\begin{bmatrix} g(x,n)(1-x) - n(-1+(1+\theta)x) -  x(1-x)\frac{\partial g}{\partial x} \\ -n(1-n)(1+\theta) \end{bmatrix}
	\end{equation}
	corresponding to $G$ and $[F,G]$, respectively. Since the second entry of $[F,G]$ is always non-zero, these vectors are linearly independent.
\end{proof}
%

\subsection{Numerical simulations}

We use the optimal control algorithm described in \cite{Hale_2016}, which is formally presented as Algorithm \ref{alg:algorithm} in the Appendix. The algorithm is based on hill climbing with Armijo step size \cite{Armijo_1966}. The direction it follows at each iteration is based on an explicit computation of the pointwise maximizer of the Hamiltonian function at time-points $t$ in a given finite grid. Hence, its effectiveness hinges on how easy it is to compute the maximizers. In the formulation \eqref{eq:control_incentive}, the state equation is affine in the control $u$ and nonlinear in the state variable. 

We applied Algorithm \ref{alg:algorithm} to the problem \eqref{eq:control_incentive}, where we fix $[R_1,S_1,T_1,P_1] = [3,1,6,2]$, $\theta = 0.7$, $x_0=0.7$, $n_0 = 0.3$, $T_f = 100$, $u_m = 1$, and Armijo parameters $\alpha=\beta = 0.5$. The values of $[R_0,S_0,T_0,P_0]$ are left unfixed in order to survey different outcomes from the distinct dynamical regions (see Figure \ref{fig:region_labels}). In addition, we leave several parameters of the algorithm to user discretion, e.g. the initial control guess $u_0(t)$ and the number of iterations. We utilize RK4-based integration solvers (ode45) for forward and backwards integration. Due to high nonlinearity of the state equation, a low error tolerance is required to produce accurate and numerically stable forward dynamics $\bm{y}(t)$, where we use relative and absolute tolerance values of $10^{-8}$.

In sample experiments from the four TOC regions, the resulting controllers were all unable to prevent the environment $n(t)$ from becoming depleted by the end of the time horizon. Of particular interest is the resulting controlled dynamics when $A_0(u)$ is confined within the V2 regime for all $u \in [-1,1]$, which we display in Figure \ref{fig:V2_incentive_u0}. With initial control guess $u_0(t) = 0$, the algorithm converges after 40 iterations to a controller that induces the states $x(t)$ and $n(t)$ to oscillate with greater amplitudes by applying the maximal negative incentive $u(t) = -1$ near the peaks of $n(t)$. It then quickly re-applies the maximum positive incentive $u(t)=1$.  Motivated by this resonance-like behavior, we run the algorithm again with the state-dependent initial guess $u_0(t) =  \text{sgn}(x-x_c)$, which switches between $\pm 1$ precisely at the points where $\dot{n}(t) = 0$. Here, $x_c \equiv 1/(1+\theta)$. The resulting controller after 20 iterations has deviated slightly away from this initial guess, and it outperforms the controller from the first experiment ($J = 25.6359$ vs $29.9707$). Also shown in Figure \ref{fig:V2_incentive_u0} are the iterates of the optimality function $\Theta(u_k)$ (see \eqref{eq:Theta} in Appendix), which is always non-positive. In these simulations, $\Theta(u_k)$ approaches zero, which indicates convergence to an optimal control satisfying PMP. The optimal controllers induce oscillatory behavior in a regime where the uncontrolled system settles at an intermediate equilibrium.

\section{Information control policies}\label{sec:control_information}

In this section, we extend the dynamics of \eqref{eq:xn_ODE} to incorporate a public opinion state  $o(t)\in[0,1]$. The state $o(t)$ is interpreted to be the average opinion in the population about the environment, and the population responds instead to this belief. We then formulate and numerically solve, using Algorithm \ref{alg:algorithm}, two optimal control problems where influence is applied directly to $o(t)$.
 
\subsection{Model with public opinion}
We introduce the following dynamics to model how opinions change in the population.
%
\begin{equation}\label{eq:xno_ODE}
	\begin{aligned}
		&\dot{x} = x(1-x)g(x,o)  \\
		&\dot{n} = n(1-n)(\theta x - (1-x)) \\
		&\dot{o} = -\gamma(o-n) \\
		&x_0,n_0,o_0\in(0,1)
	\end{aligned}
\end{equation}
where $\gamma > 0$. The form of the $\dot{o}$ equation induces $o(t)$ to track the environmental state $n(t)$. There is a lag between actual changes in the environment and the public becoming informed about the changes. The learning parameter $\gamma > 0$ determines how slow this lag is. For $\gamma$ low, $o(t)$ will not adapt quickly to the fluctuating $n(t)$. As $\gamma$ increases, $o(t)$ more successfully tracks $n(t)$. The $\dot{x}$ equation above is modified from \eqref{eq:xn_ODE} by replacing the relative fitness $g(x,n)$ with $g(x,o)$. Here, individual incentives are now determined by the current public opinion and not the true environmental state $n$. Thus, the previous system dynamics \eqref{eq:xn_ODE} can be interpreted as the population responding to perfect information about the environment, $o(t) = n(t)$ $\forall t\geq 0$. We denote the system mapping \eqref{eq:xno_ODE} with the mapping $F^o: [0,1]^3 \rightarrow \mbb{R}^3$. 
\begin{figure}
	\centering
	\includegraphics[scale=.27]{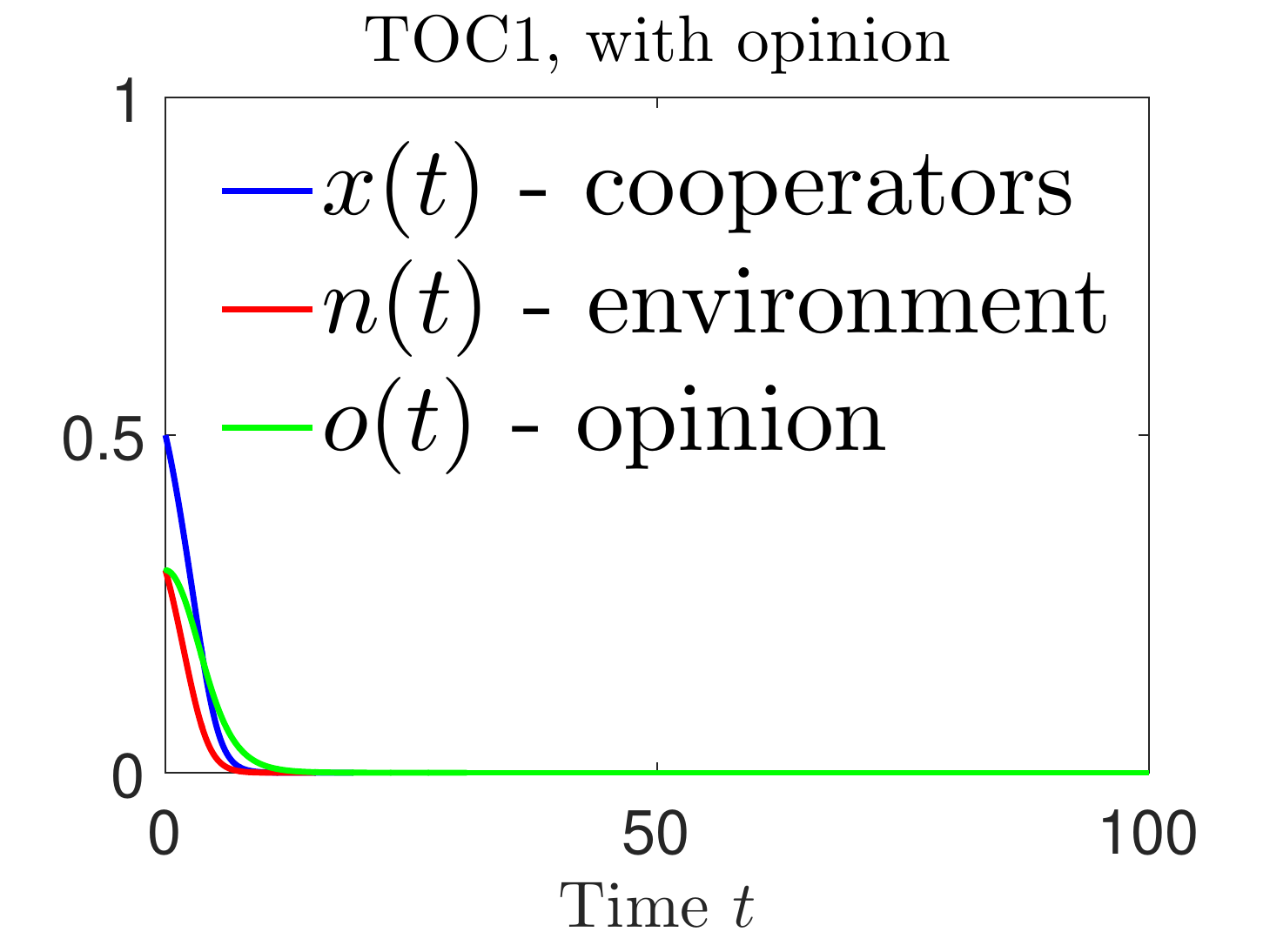}
	\includegraphics[scale=.27]{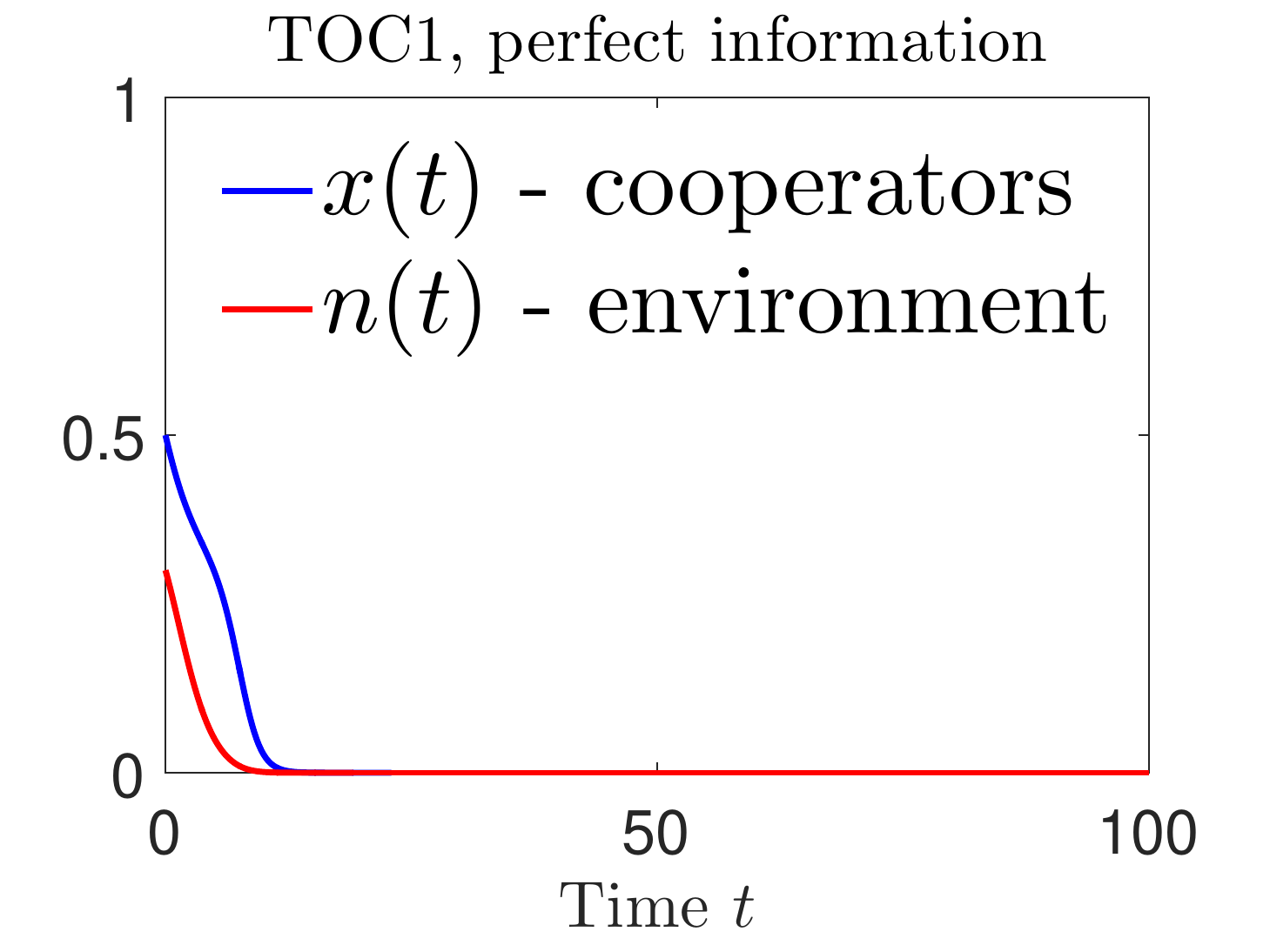}\vspace{2mm}
	\includegraphics[scale=.27]{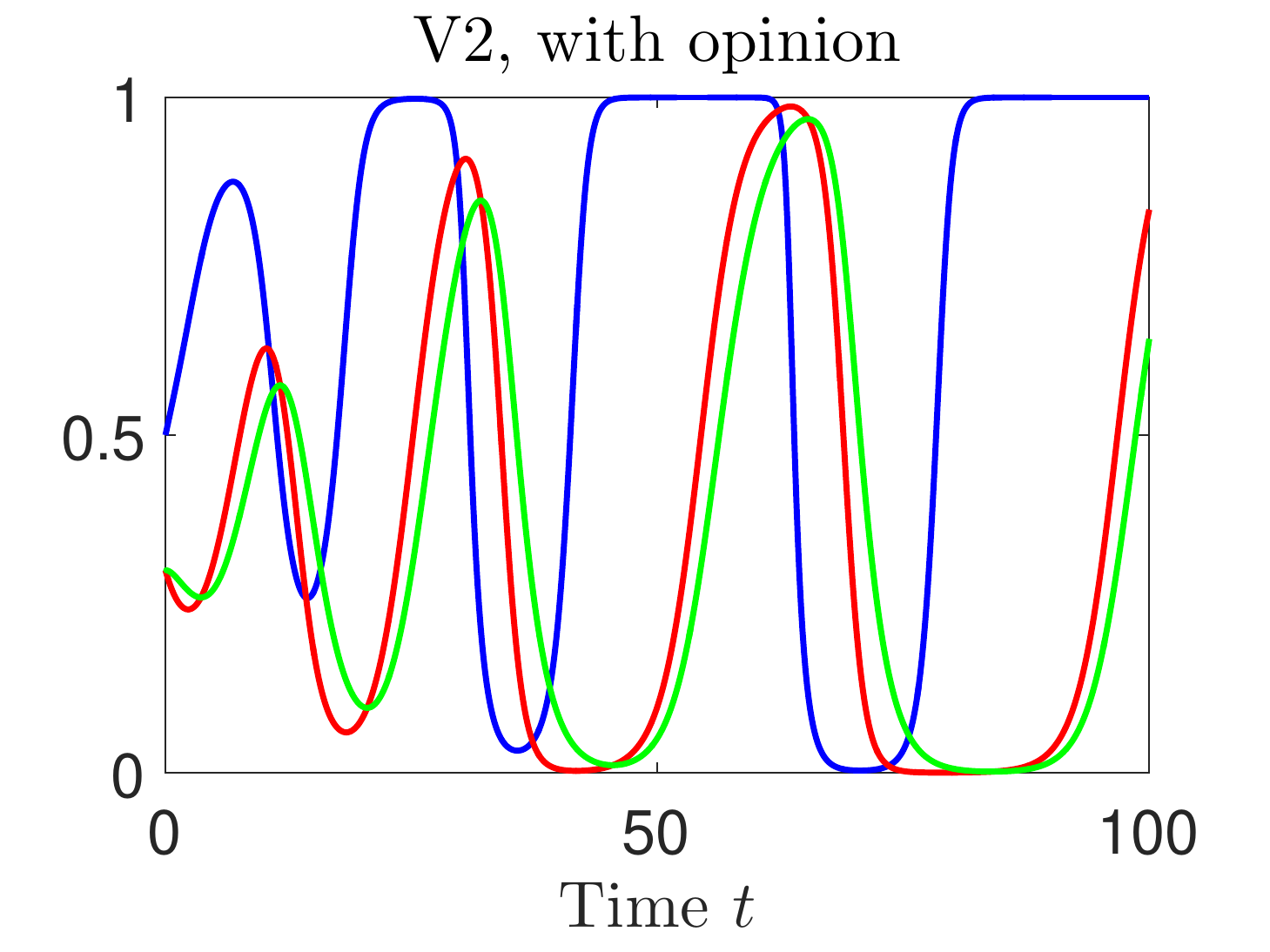}
	\includegraphics[scale=.27]{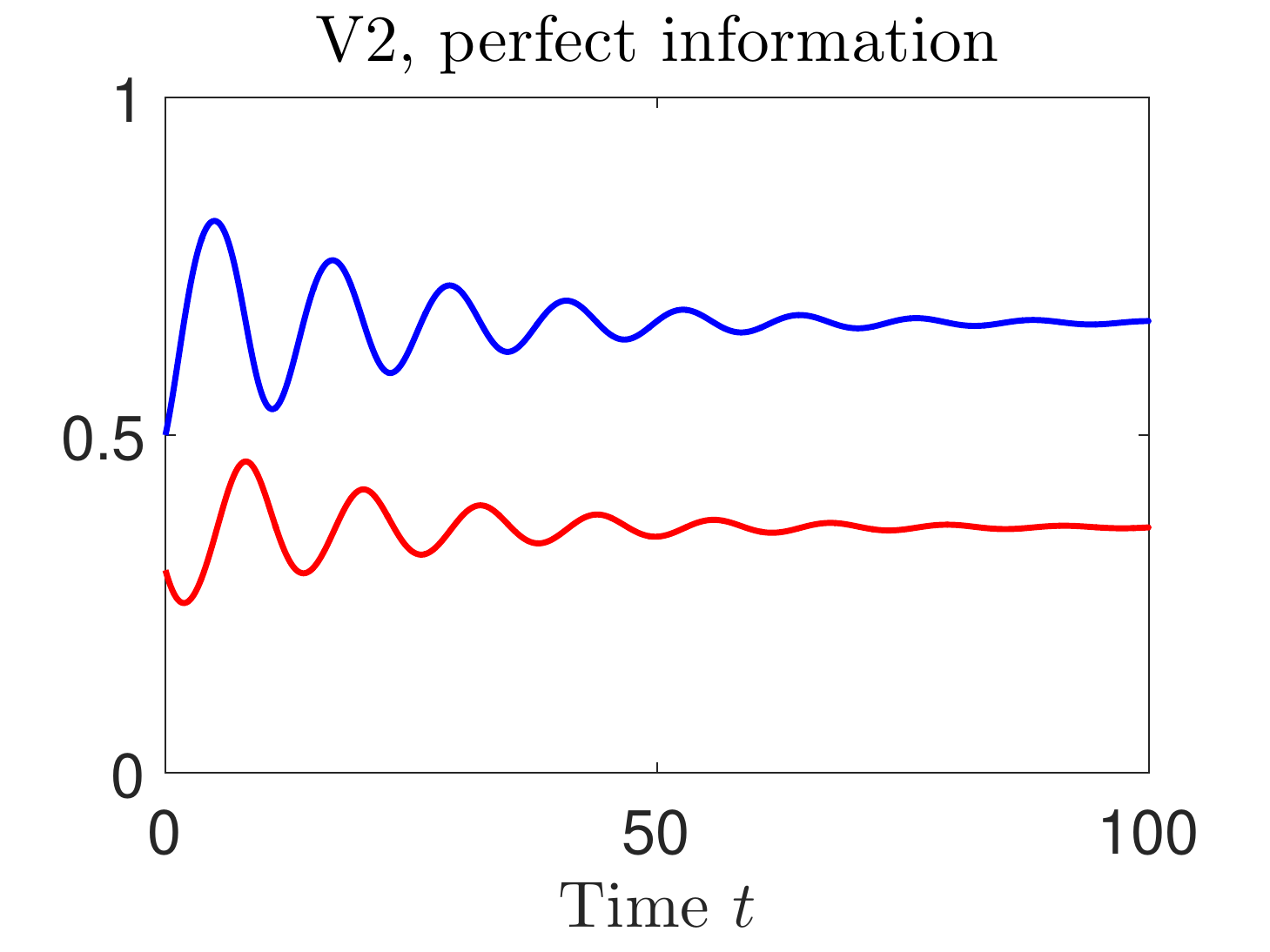}\vspace{2mm}
	\includegraphics[scale=.27]{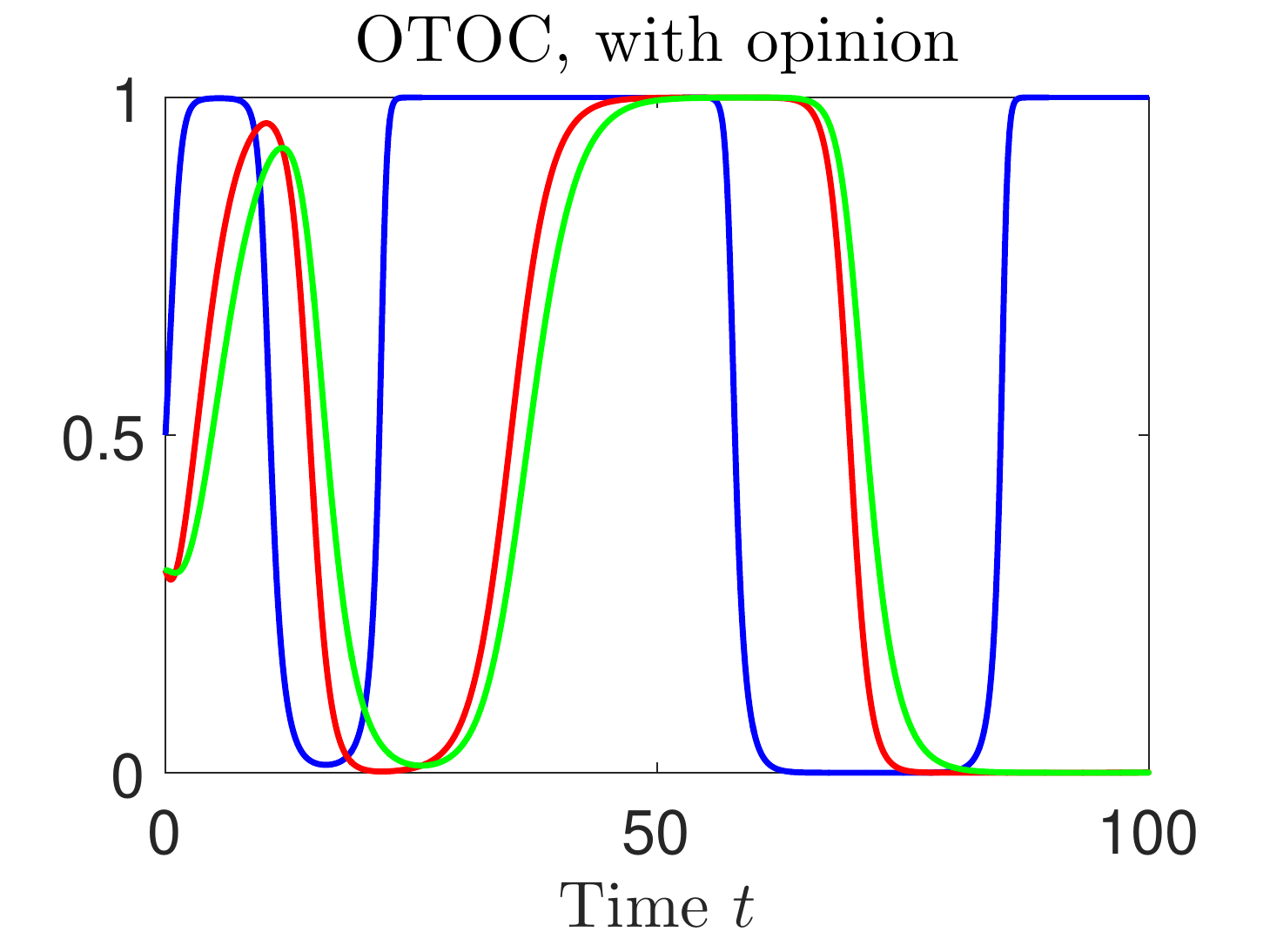}
	\includegraphics[scale=.27]{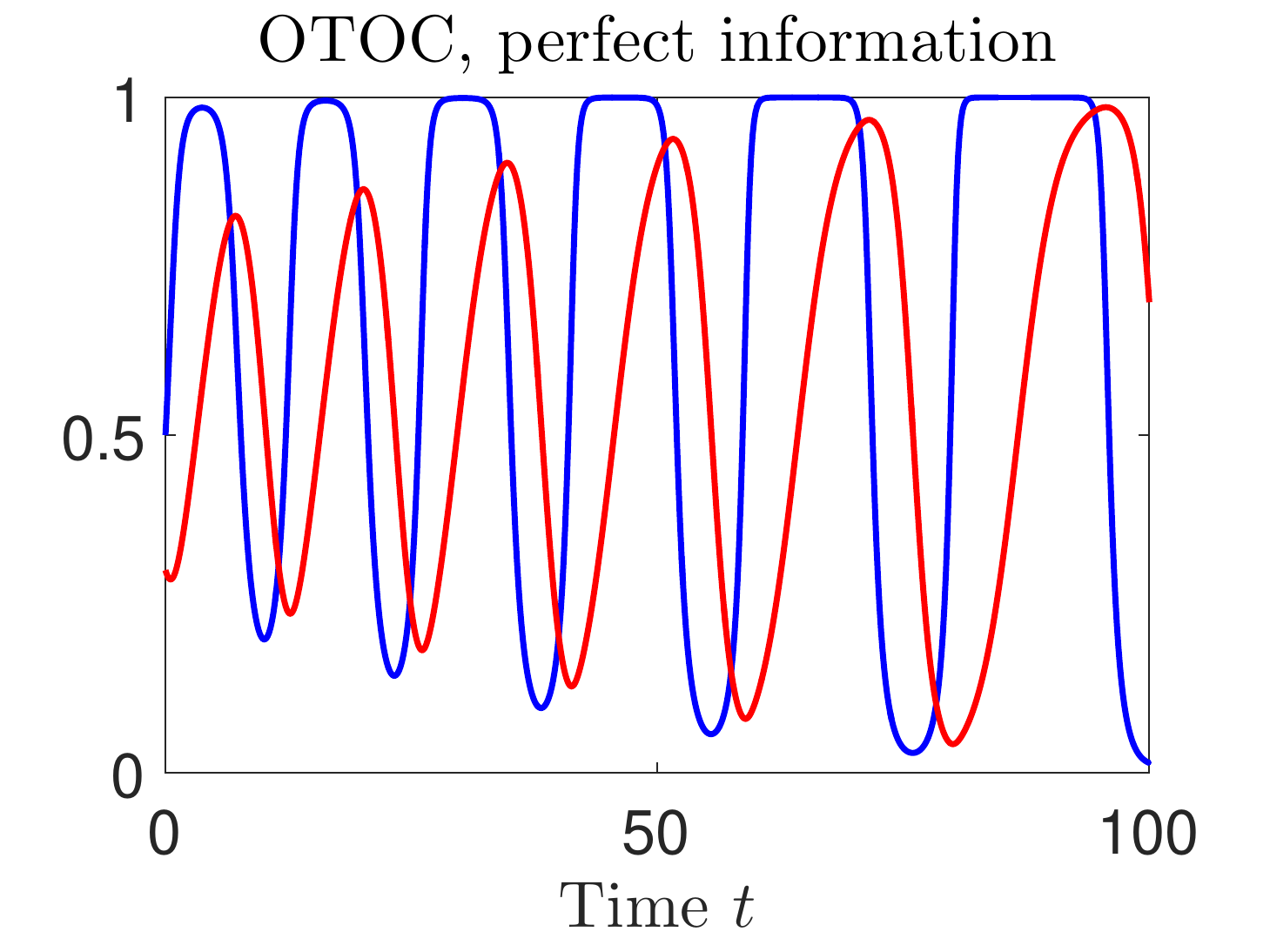}\vspace{2mm}
	\caption{Comparison of public opinion-induced dynamics (Left column) and the original feedback-evolving game (Right column). (Top row) $A_0=[R_0,S_0;T_0,P_0] = [5,2;3,3]$ (regime TOC1). Delay of opinion does not help to restore the commons. (Middle row) $A_0 = [4.5,4;3,3]$ (regime V2). Delayed opinion destabilizes the interior fixed point. (Bottom row) $A_0 = [7,4;3,3]$ (regime OTOC). Public opinion facilitates convergence to heteroclinic cycle in the $(x,n)$ trajectories. In all simulations, $[R_1,S_1,T_1,P_1] = [3,1,6,2]$, $\gamma = .5$, $\theta = .5$, $x_0 = .5$, $n_0 = .3$, $o_0 = .3$.  }
	\label{fig:opinion_comparison}
\end{figure}

We illustrate the dynamical effects of the public opinion for three dynamical regimes in Figure \ref{fig:opinion_comparison}. A notable effect occurs in the V2 and OTOC regimes, where the $(x(t),n(t))$ trajectories are pushed towards the boundary of the state space. This is due to the delay in opinion, and the intuition is as follows. When $n(t)$ starts to increase towards a peak, $o(t)$ lags behind and stays below $n(t)$. This causes more of the population to become cooperators, since they are responding to lower public opinion relative to the true resource state. As a result, $n(t)$ is restored more than it would have been if the population had perfect information. Then, $o(t)$ overestimates $n(t)$ as it decreases, causing more of the population to defect, degrading the environment. This process continues to repeat, causing oscillations to have larger amplitudes. 

\subsection{Optimal control formulation: propaganda strategies}
Here, we consider an external entity, e.g. media platforms, politicians, and activists, that seeks to maximally conserve the environment by influencing the public's opinion. First, we study policies that perturb opinion by injecting information. Propaganda and media broadcasts can achieve such perturbations, for example. We formulate the following optimal control problem.
\begin{equation}\label{eq:control_propaganda}
	\begin{aligned}
		&\max_{u} J=  \frac{1}{2}\int_0^{T_f} C_1 n^{2}(t) - C_2 u^2(t) dt \\
		&\text{subject to } \begin{cases} \dot x = x(1-x)g(x,o) \\
		\dot n =n(1-n)(-1+(1+\theta)x) \\
		\dot o = -\gamma(o-n) + o(1-o)u \\
		x_0,n_0,o_0 \in (0,1)  \end{cases}
	\end{aligned}
\end{equation}
where $C_1,C_2> 0$ are the priority and regulator weights, respectively. We denote the above dynamics as $\dot{\bm{y}} = F^o(\bm{y},u)$. The additive control term $o(1-o)u$ serves two purposes. First, it keeps the dynamics well-posed, i.e. a solution $o(t)$ that starts in $[0,1]$ will stay in $[0,1]$. Second, it models the difficulty to influence extreme opinions. The additive term decreases to zero as $o$ approaches the extremes 0 and 1, and hence more influence is required to move $o(t)$ away from the extremes. Note that the control function is left unconstrained, $u(t) \in \mbb{R}$ $\forall t \in [0,T_f]$.
\begin{figure*}[t!]
	\centering
	\begin{subfigure}[t]{\columnwidth}
		\hspace{-5mm}\includegraphics[scale=.3]{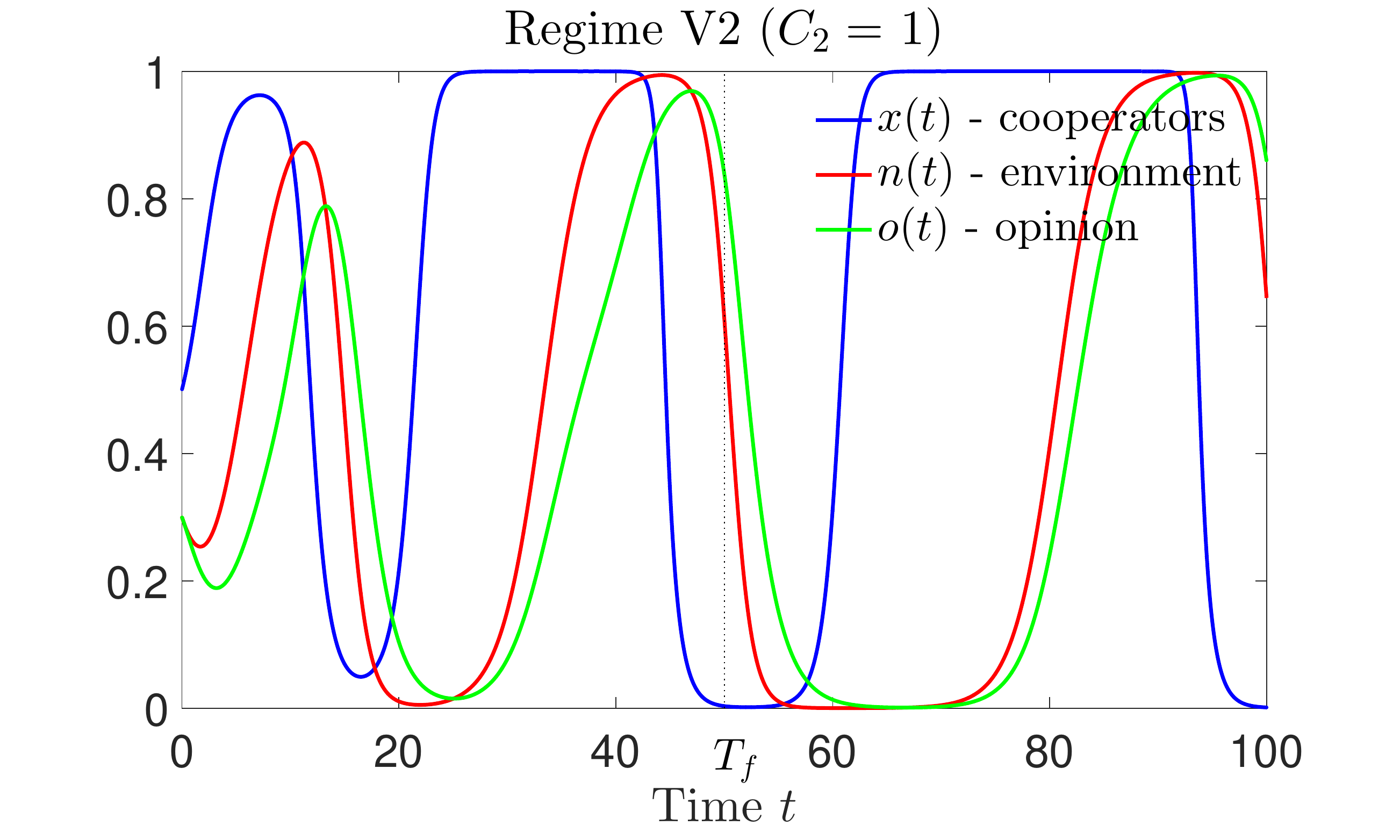}
		\includegraphics[scale=.25]{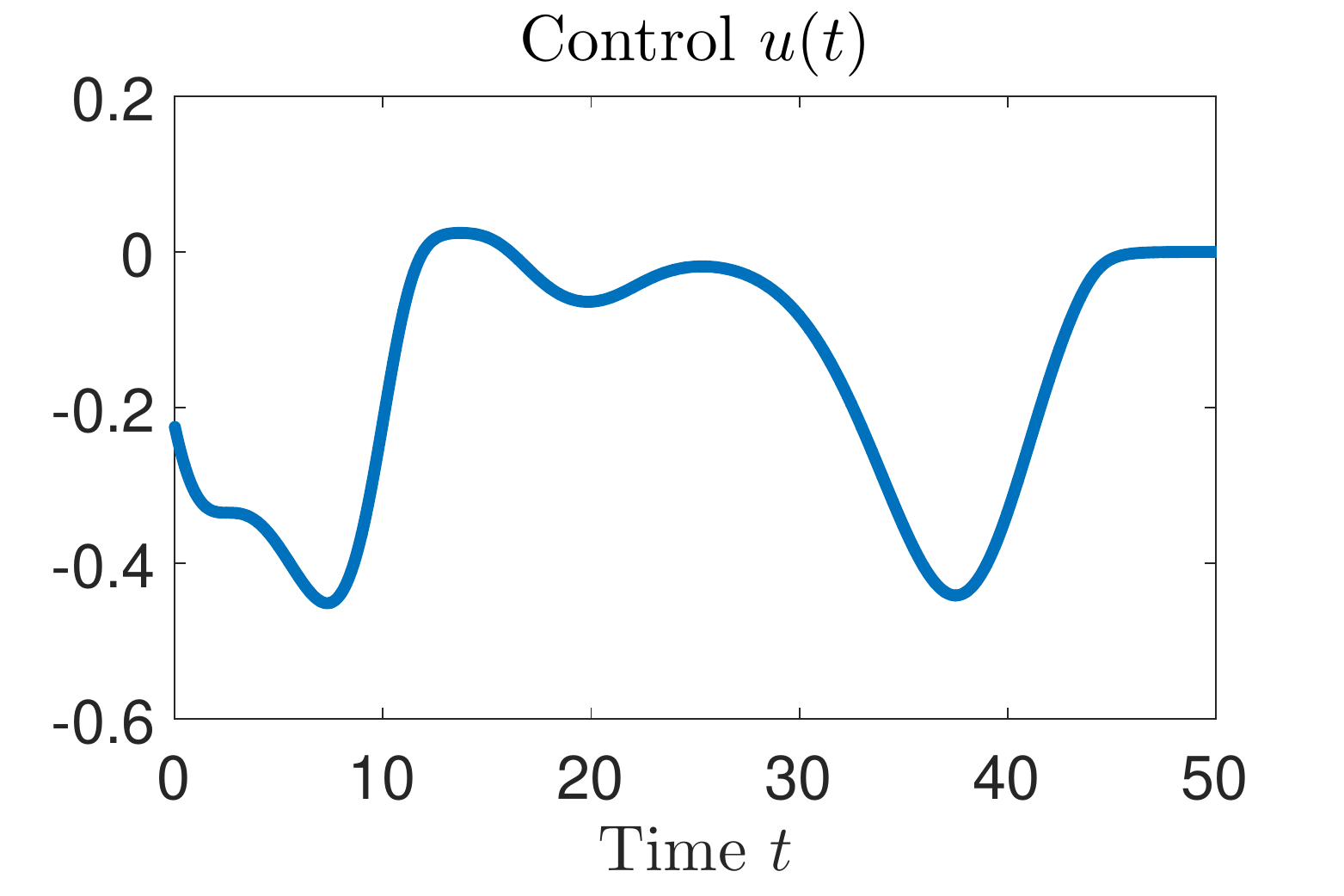}
		\includegraphics[scale=.25]{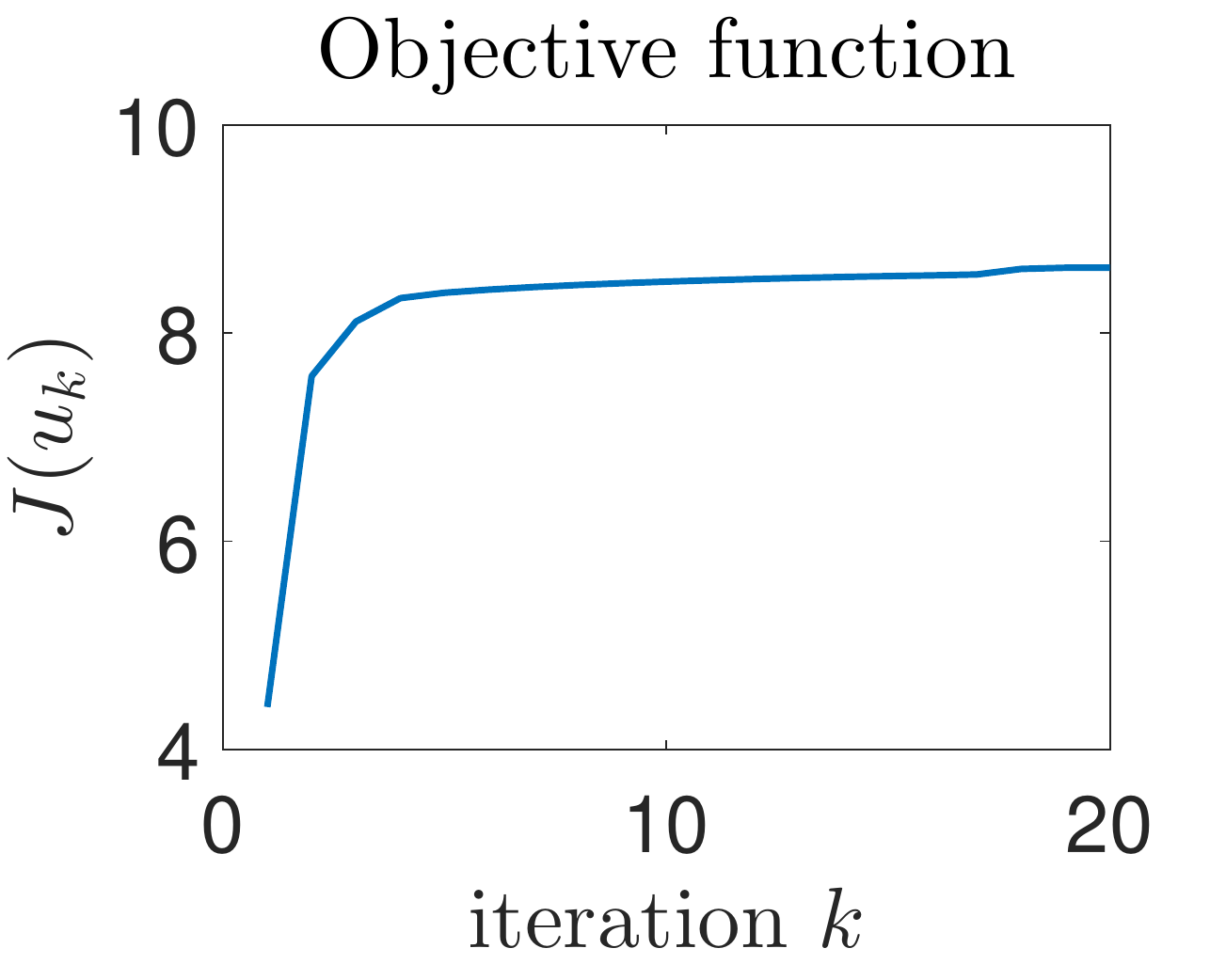} 
		\centering
		\caption{}
		\label{fig:propaganda_V2_C2_1}
	\end{subfigure}  
	\begin{subfigure}[t]{\columnwidth}
		\hspace{-5mm}\includegraphics[scale=.3]{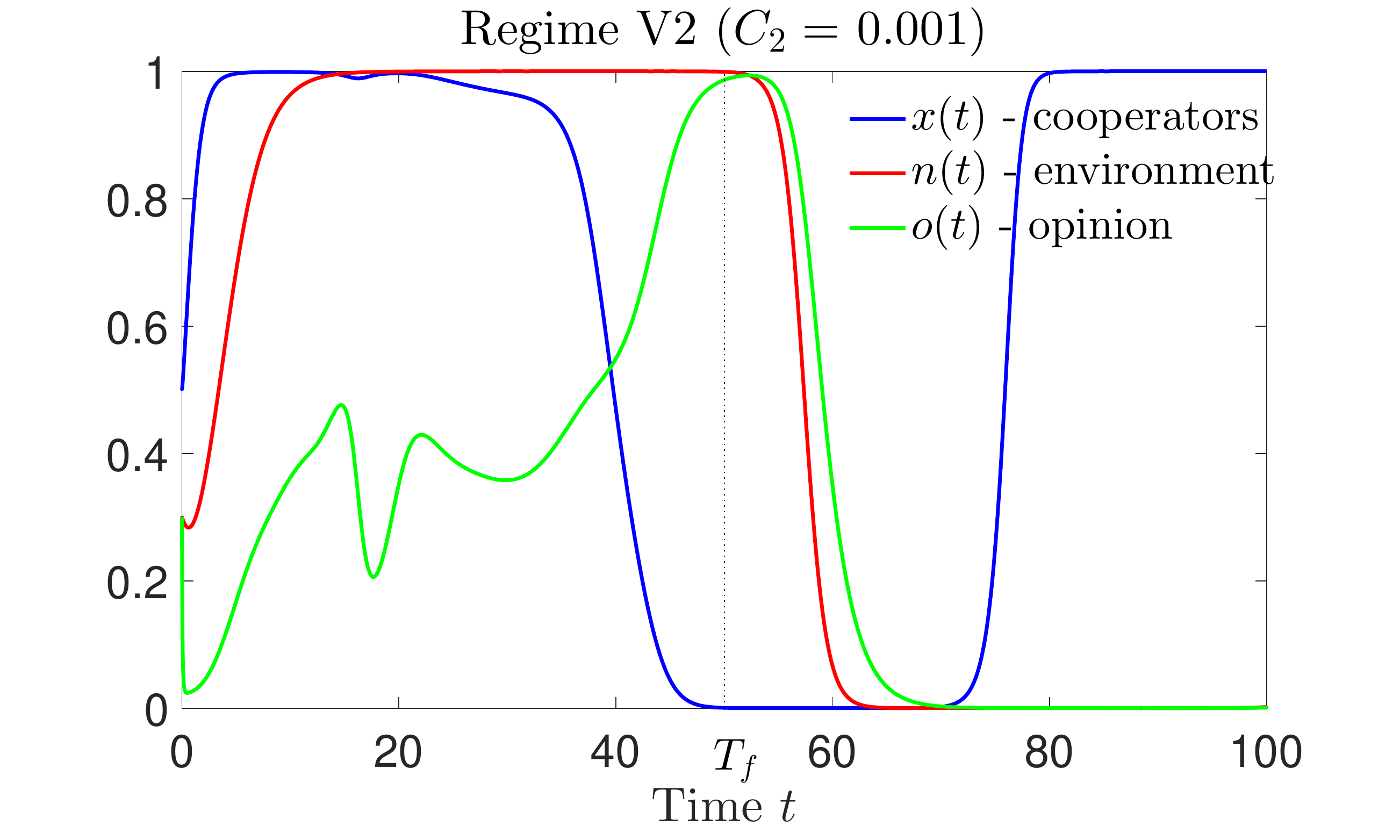}
		\includegraphics[scale=.25]{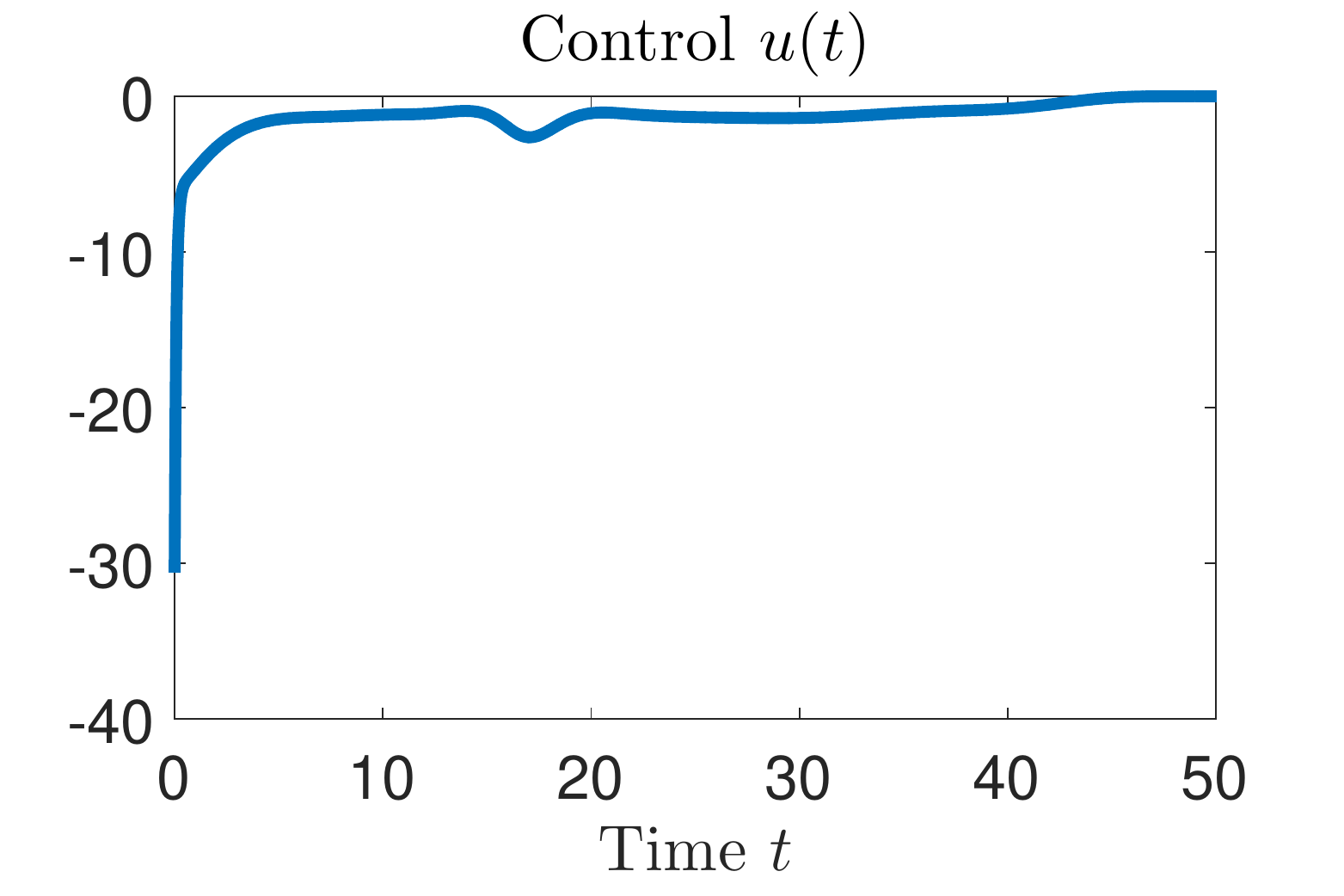}
		\includegraphics[scale=.25]{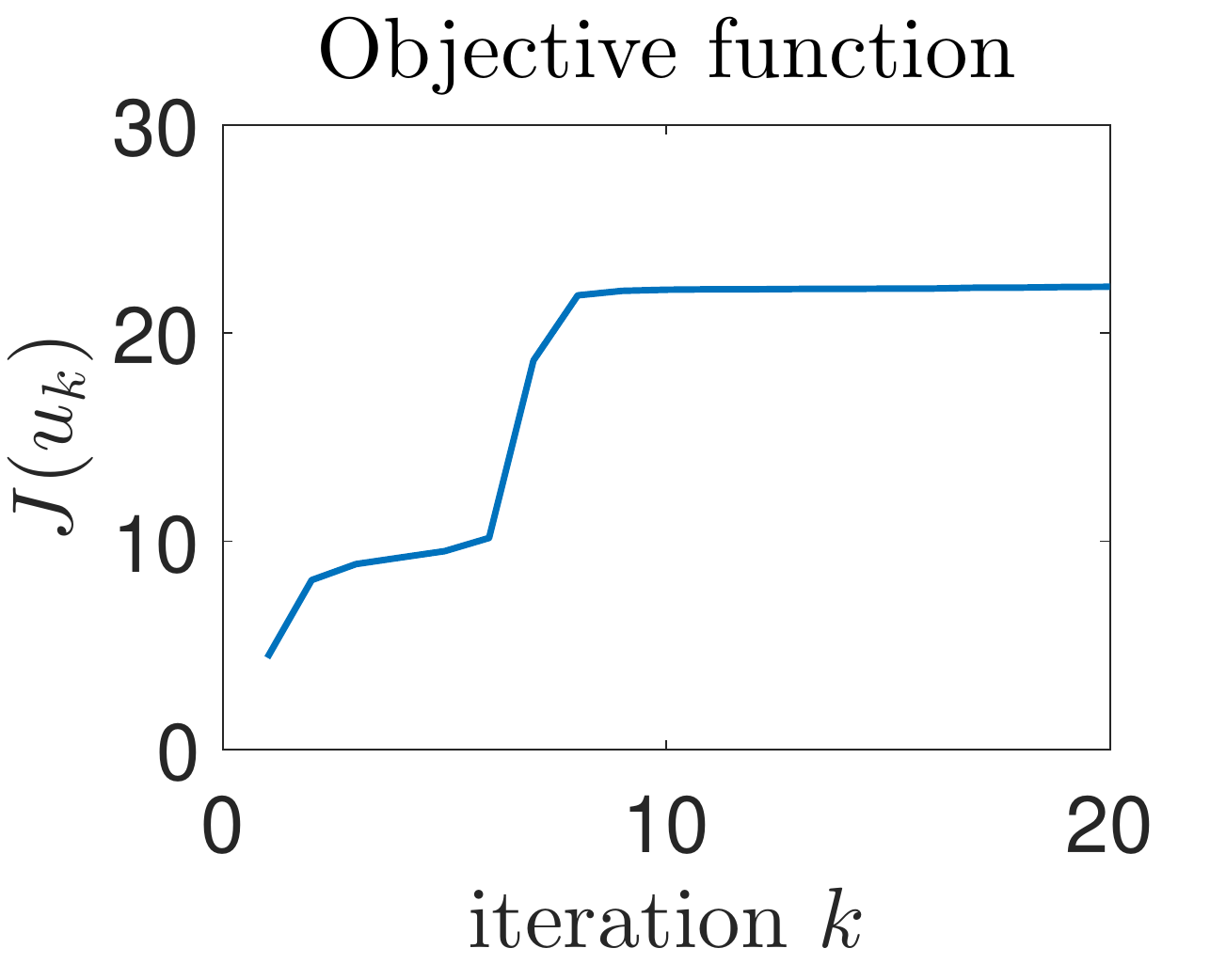}
		\centering
		\caption{}
		\label{fig:propaganda_V2_C2_point001}
	\end{subfigure} 
	\caption{An application of Algorithm \ref{alg:algorithm} with $[R_0,S_0,T_0,P_0] = [4.5,4,3,3]$ (V2 regime) to propaganda control problem \eqref{eq:control_propaganda} with $T_f = 50$. In left panels (a), we applied 20 iterations (runtime 69.938 s) with $u_0(t) = 0$. (Top) State trajectories. After $u_{20}(t)$ is applied on $t\in[0,T_f]$, dynamics are continued without control for a time of length 50. (Bottom Left) The control function $u_{20}(t)$. (Bottom Right) Objective scores $J(u_k)$ vs iteration number $k$, where $J(u_{20}) = 8.629$ and $\Theta(u_{20}) =  -0.0048$ (not plotted).  In right panels (b), we set $C_2 = 0.001$ and run 20 iterations (runtime 257.197 s). We obtain $J(u_{20}) =22.22$ and $\Theta(u_{20}) =  -0.023$.}
\end{figure*}
The Hamiltonian is
\begin{equation}
	\begin{aligned}
		H(\bm{y},\bm{\lambda},u) &= \lambda_x x(1-x)g(x,o) + \lambda_n n(1-n)(\theta x - (1-x)) \\
		& +\lambda_o(-\gamma(o-n) + o(1-o)u) +\frac{1}{2} (C_1 n^2 - C_2 u^2)
	\end{aligned}
\end{equation}
where the costate $\bm{\lambda} = [\lambda_x,\lambda_n,\lambda_o]^\top$ obeys the dynamics
\begin{equation}
	\dot{\bm{\lambda}} = -\frac{\partial H}{\partial \bm{y}}(\bm{y},\bm{\lambda},u)
\end{equation}
with $\bm{\lambda}(T_f) = [0,0,0]^\top$. The expression of $H$ is concave in $u$, and hence it admits the unique point-wise maximizer
\begin{equation}\label{eq:ustar_propaganda}
	u^*(t) = \frac{1}{C_2}\lambda_o(t) o(t)(1-o(t)).
\end{equation}

We applied Algorithm \ref{alg:algorithm} to the problem \eqref{eq:control_propaganda}. We fix the priority weight $C_1 = 1$, and study modifications to the regulator weight $C_2$. We fix $[R_1,S_1,T_1,P_1] = [3,1,6,2]$, $\theta = 0.5$, $\gamma = 0.5$, $(x_0,n_0,o_0) = (.5,.3,.3)$, $T_f = 50$, and Armijo parameters $\alpha=\beta = 0.5$. First, a notable observation was that the environmental state could be rescued in the TOC1 regime for a limited time, followed by collapse, if effort cost was low ($C_2 = .001$). When the cost weights are balanced, e.g. $C_2 = 1$, we did not observe resurgence of the commons in any of the TOC  regimes.  

In regime V2, when control effort is balanced ($C_2 = 1$, Figure \ref{fig:propaganda_V2_C2_1}), the computed control applies effort in waves. The control starts with a high negative amplitude as $o(t)$ and $n(t)$ begin to ascend, pushing $o(t)$ lower relative to $n(t)$ to cause a resurgence of cooperators (blue line). The control then relaxes its effort as the states $o(t),n(t)$ begin to decrease, $t\approx 15$. It applies negative effort again as $n(t)$ and $o(t)$ begin to ascend around $t\in[30,40]$. Hence, the control $u(t)$ promotes cooperation through negative control effort at selected times during the horizon. This causes the environment to oscillate between more extreme depleted and repleted states. In Figure \ref{fig:propaganda_V2_C2_point001} with $C_2 = 0.001$, control effort is cheap. The resulting control applies a large negative impulse at the beginning to push $o(t)$ very low. This stimulates the growth of cooperators and consequently, the environment, which stays near $n=1$ until after $t=T_f$. After the initial impulse, $u(t)$ relaxes for the rest of the horizon, causing $o(t)$ to eventually catch up to $n(t)$ and causing defectors to dominate. In the absence of control (after $T_f$), the environment collapses but will be subject again to another resurgence. Similar results are obtained when applying the algorithm in the OTOC dynamical regime.

\begin{figure*}[t!]
	\centering
	\begin{subfigure}[t]{\columnwidth}
		\hspace{-5mm}\includegraphics[scale=.3]{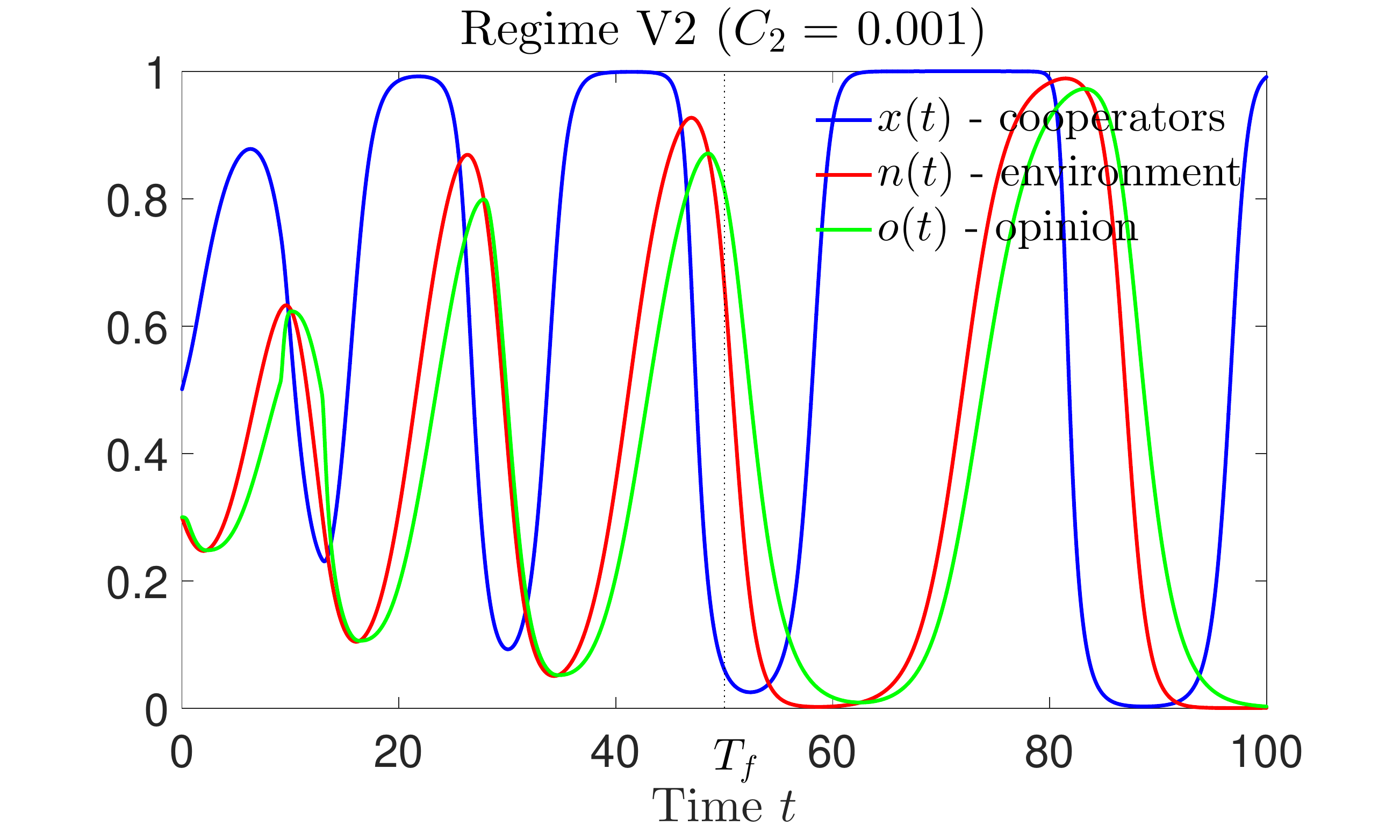}
		\includegraphics[scale=.25]{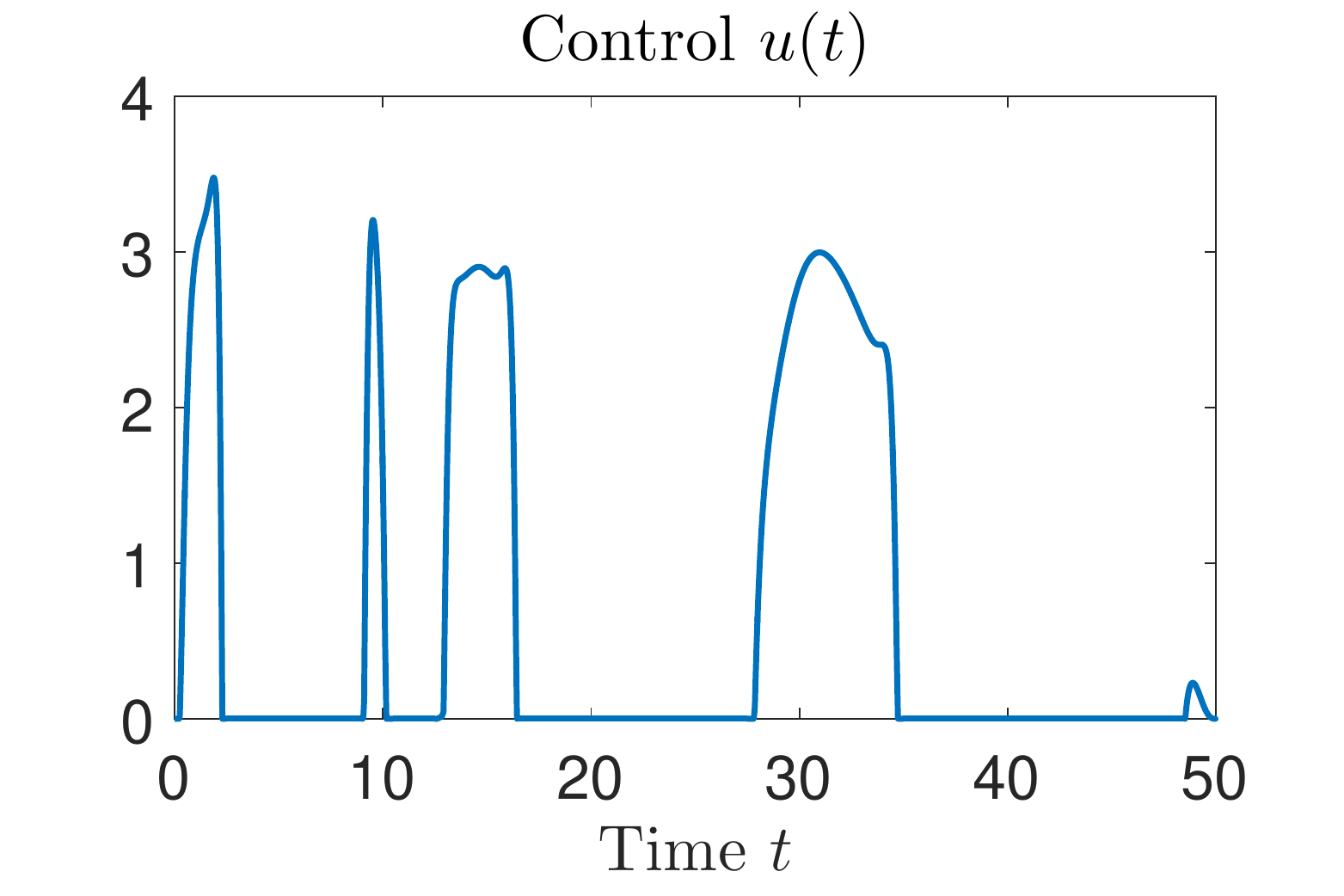}
		\includegraphics[scale=.25]{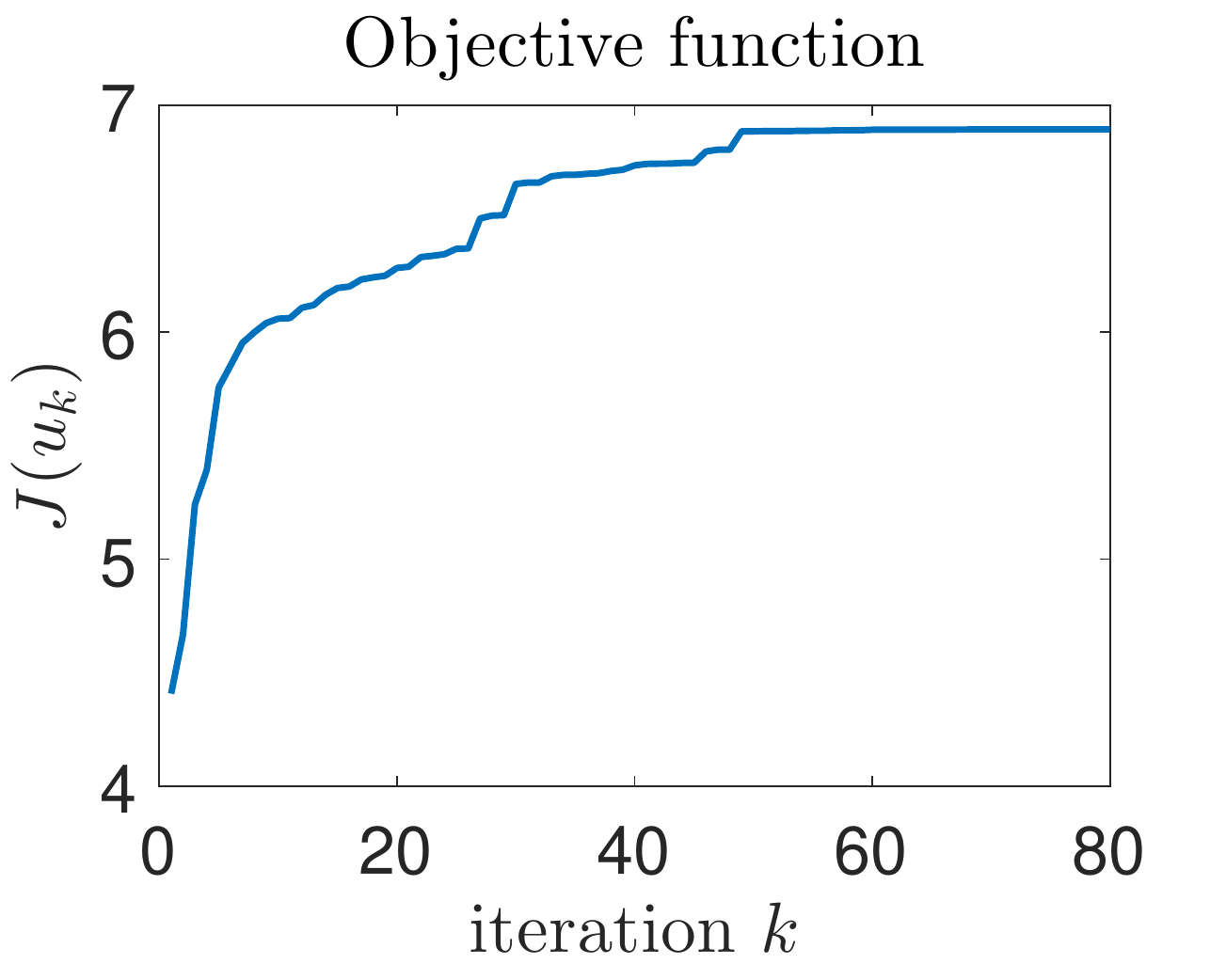} 
		\centering
		\caption{V2 regime with $[R_0,S_0,T_0,P_0] = [4.5,4,3,3]$. }
		\label{fig:awareness_V2_C2_point001}
	\end{subfigure}  
	\begin{subfigure}[t]{\columnwidth}
		\hspace{-5mm}\includegraphics[scale=.3]{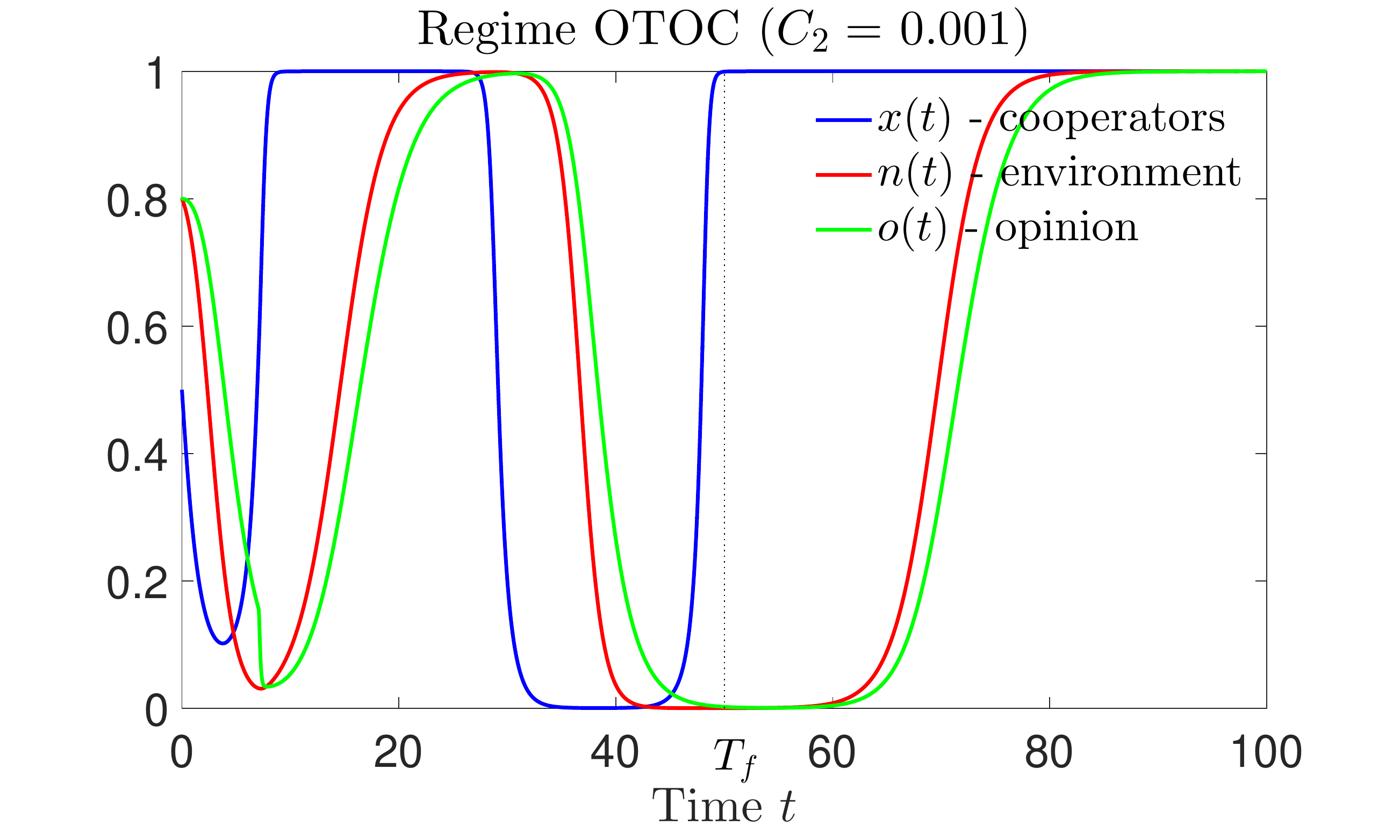}
		\includegraphics[scale=.25]{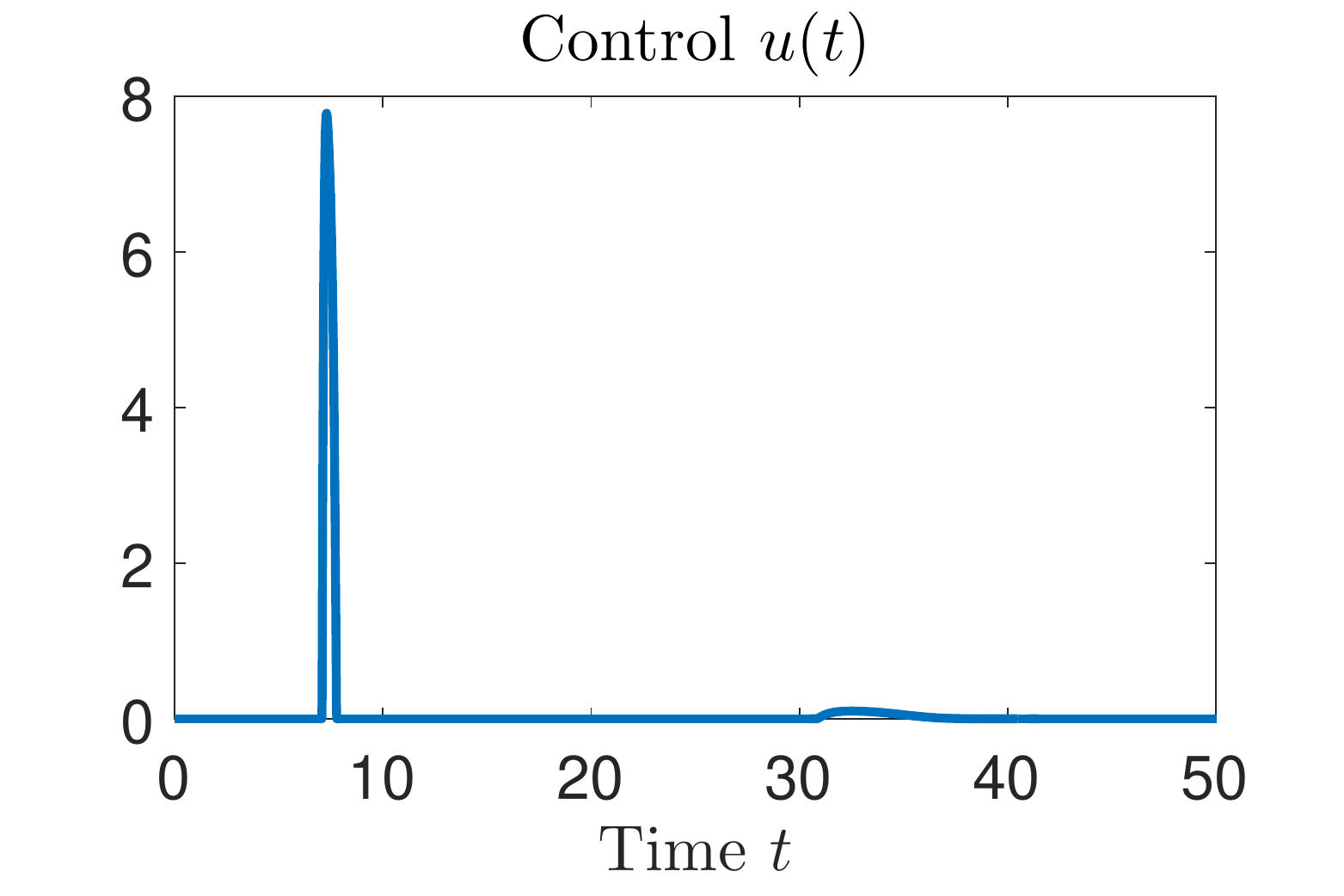}
		\includegraphics[scale=.25]{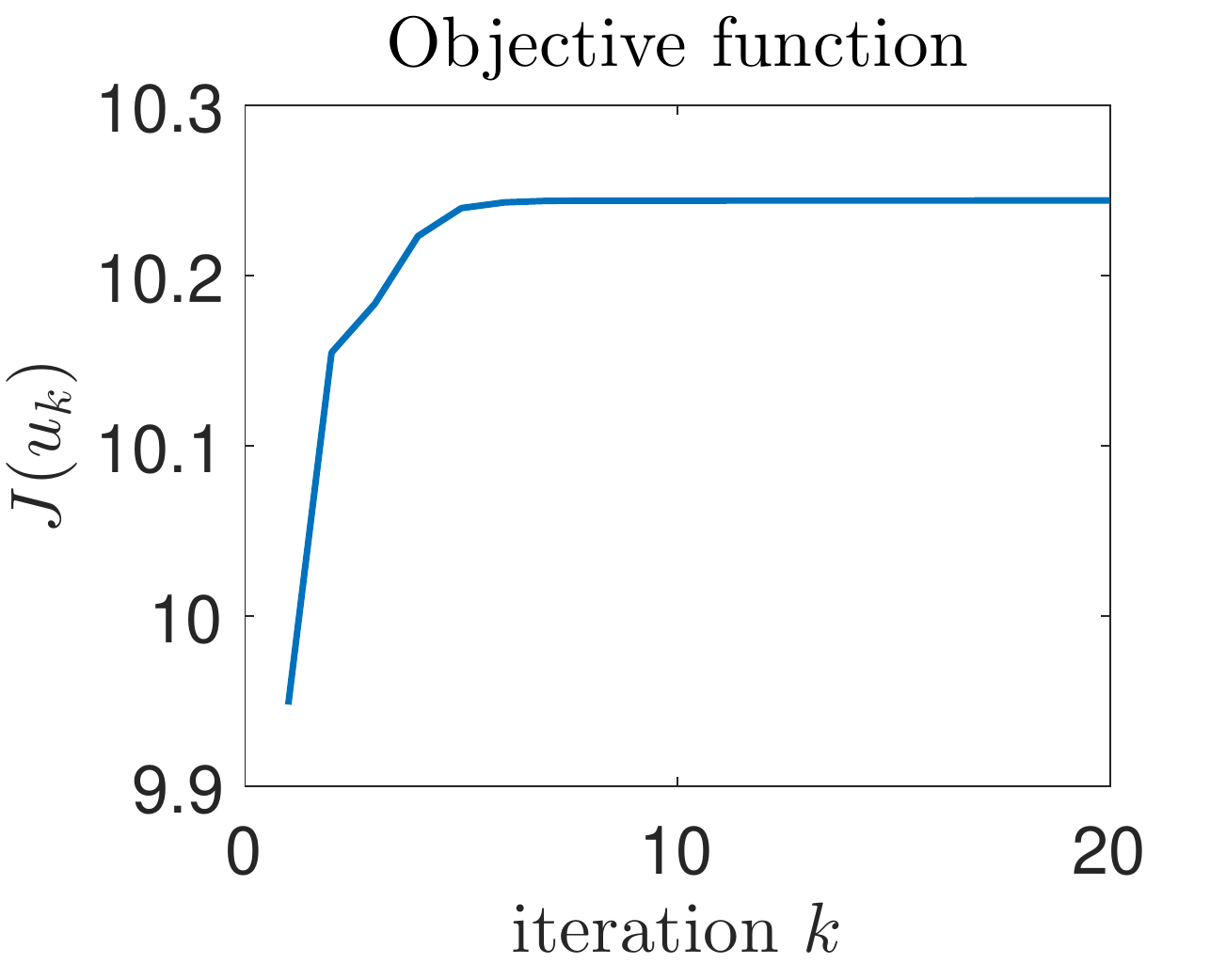}
		\centering
		\caption{OTOC regime with $[R_0,S_0,T_0,P_0] = [7,4,3,3]$. }
		\label{fig:awareness_OTOC_C2_point001}
	\end{subfigure} 
	\caption{An application of Algorithm \ref{alg:algorithm} to the awareness control problem \eqref{eq:control_awareness} with $T_f = 50$ to compute an optimal control $t\in[0,T_f]$. In left panels (a), we applied 80 iterations (runtime 313.2 s) in the V2 regime with $u_0(t) = 0$. (Top) State trajectories. After control is applied on $t\in[0,T_f]$, the dynamics are continued without control for a time of length 50. (Bottom Left) The control $u_{20}(t)$ after 20 iterations. (Bottom Right) Objective scores $J(u_k)$ vs iteration number $k$, where $J(u_{20}) = 6.894$ and $\Theta(u_{20}) =  -0.0001$.  In right panels (b), we set $C_2 = 0.001$ and run 20 iterations (runtime 51.87 s) with $x_0=0.5$, $n_0 = o_0 = 0.8$, and $A_0$ in the OTOC regime. We obtain $J(u_{20}) =10.24$ and $\Theta(u_{20}) =  -3.271\times 10^{-6}$.}
\end{figure*}
%

\subsection{Environmental awareness strategies}
We now consider strategic information policies that guide public opinion towards the true environmental state $n(t)$. Environmental awareness and educational campaigns are examples of interventions that serve this purpose. We formulate the following optimal control problem, with $u(t) \in [0,\infty)$ non-negative for all $t \in [0,T_f]$ directly affecting the public learning parameter $\gamma > 0$.
\begin{equation}\label{eq:control_awareness}
	\begin{aligned}
		&\max_{u} J=  \frac{1}{2}\int_0^{T_f} C_1 n^{2}(t) - C_2 u^2(t) dt \\
		&\text{subject to } \begin{cases} \dot x = x(1-x)g(x,o) \\
		\dot n =n(1-n)(-1+(1+\theta)x) \\
		\dot o = -(\gamma+u)(o-n) \\
		u(t) \in [0,\infty), \forall t \in [0,T_f] \\
		x_0,n_0,o_0 \in (0,1)  \end{cases}
	\end{aligned}
\end{equation}
The Hamiltonian of the awareness control problem is
\begin{equation}
	\begin{aligned}
		H(\bm{y},\bm{\lambda},u) &=  \lambda_x x(1-x)g(x,o) + \lambda_n n(1-n)(\theta x - (1-x)) \\
		& \ \ \ -\lambda_o(\gamma + u)(o-n) +\frac{1}{2} (C_1 n^2 - C_2 u^2).
	\end{aligned}
\end{equation}
where the costate $\bm{\lambda} = [\lambda_x,\lambda_n,\lambda_o]^\top$ obeys the dynamics
\begin{equation}
	\dot{\bm{\lambda}} = -\frac{\partial H}{\partial \bm{y}}(\bm{y},\bm{\lambda},u)
\end{equation}
with $\bm{\lambda}(T_f) = [0,0,0]^\top$. The pointwise maximizer of $H$ under the non-negativity constraint on $u(t)$ is
\begin{equation}\label{eq:ustar_awareness}
	u^*(t) = \begin{cases} 0 \ &\text{if} -(1/C_2)\lambda_o(o-n) < 0 \\ -(1/C_2)\lambda_o(o-n) &\text{if } -(1/C_2)\lambda_o(o-n) \geq 0 \end{cases}
\end{equation}

An application of Algorithm \ref{alg:algorithm} to the V2 regime is shown in Figure \ref{fig:awareness_V2_C2_point001}, where we set $C_2 = 0.001$, $\gamma = 0.5$, $\theta = 0.5$, and $[R_1,S_1,T_1,P_1] = [3,1,6,2]$. Due to the lag $\gamma$, public opinion $o(t)$ overestimates $n(t)$ on the intervals where $n(t)$ is decreasing, and underestimates when $n(t)$ is increasing. The resulting awareness control is applied only during these intervals to push opinion lower towards the true environmental state, and is not applied ($u(t) = 0$) on intervals where $n(t)$ is increasing. Consequently, the awareness  policy promotes cooperative behavior in times where public opinion overestimates the true environmental state. The induced dynamics resemble an oscillating tragedy of the commons.

In Figure \ref{fig:awareness_OTOC_C2_point001}, a similar principle holds for the resulting controller in the OTOC regime. A single impulse of awareness control is applied around $t=7$, when $o(t)\approx 0.16$ greatly overestimates $n(t) \approx 0.03$. Opinion quickly decreases to meet $n(t)$, causing a resurgence of cooperators. No more control is applied for the rest of the horizon. 


\section{Conclusions and discussion}\label{sec:conclusion}

In this paper, we extended  a game-environment feedback model \cite{Weitz_2016}  to study incentive and opinion control policies that seek to maximally conserve the environmental state. We formulated these policies in the setting of optimal control problems, and solved them by using suitable numerical techniques. The computed incentive policies are bang-bang controllers that, counter-intuitively, switch between maximal promotion and punishment of cooperative behaviors. The switching times occur near critical points of the environmental state dynamics. We then considered two methods of influencing public opinion about the environment. The first is a propaganda-like intervention where an external influencing agent attempts to sway public opinion. The second aims to raise public awareness of the current true environmental state, e.g. through environmental education programs or awareness campaigns. In  simulations, both methods steer public opinion lower, i.e. convincing the public that the environment is worse than what it actually is. 

We find in certain regimes (V2 and OTOC)  that the resulting controllers in all three control formulations induce large oscillations between deplete and replete environment states. The resulting oscillating tragedy of the commons maximizes accumulation of common resources because the policies increase the amount of time spent at high replete states. The major drawback is that repeated collapses of the resource are inevitable. This outcome is extremely undesirable if there are no alternative resource options. Hence, different ways of thinking about control are necessary.

\appendix
For all three problem formulations \eqref{eq:control_incentive}, \eqref{eq:control_propaganda}, and \eqref{eq:control_awareness}, we utilized the optimal control algorithm outlined below in Algorithm \ref{alg:algorithm}.
 For the interested reader, the details can be found in \cite{Hale_2016}.
 The algorithm is a hill-climbing technique with Armijo step sizes \cite{Armijo_1966}. Given a control
$u$, it computes an ascent direction, $u^{\star}$, as follows. First, choose a finite grid ${\mcal G}\subset [0,T_f]$, which may vary from one iteration to the next (in this paper, we fix the grid to have uniform spacing of 0.01). Solve the state trajectory forward and the costate (adjoint) trajectory backwards, by a numerical integration method. This yields the Hamiltonian function $H(\bm{y},\bm{\lambda},u)$. For every $t\in{\mcal G}$, compute the maximizer $u^*(t)$ of $H(\bm{y}(t),\bm{\lambda}(t),v)$
over admissible controls $v$. Interpolate the resulting values via zero-order hold to result in the control $u^{\star}(t)$ for every $t\in[0,T_f]$. The control $u^{\star}$ serves as the direction the algorithm takes from $u$. The cost functional $J$ increases along this direction, namely,
for a small enough step size $\delta>0$,
\begin{equation}
	J(u+\delta(u^{\star}-u))>J(u)
\end{equation}
under mild technical conditions.  In the algorithm, we use the Armijo step size, which is computed as follows. For a given $\beta\in(0,1)$, and a given $\alpha\in(0,1)$, the step size is $\beta^{\ell}$
where $\ell$ is the smallest non-negative integer such that
$J(u)-J(u+\beta^{\ell}(u^{\star}-u))\leq\alpha\beta^{\ell}\Theta(u)$, where 
\begin{equation}\label{eq:Theta}
	\Theta(u) \equiv \int_0^{T_f} (  H(\bm{y},\bm{\lambda},u) - H(\bm{y},\bm{\lambda},u^*) ) dt \leq 0.
\end{equation}
The term $\Theta(u)$ serves as an optimality function (see {\cite{Polak_1997}): It is always non-positive, where
$\Theta(u)=0$ means that $u$ satisfies PMP. Generally, $|\Theta(u)|$ measures the extent to which 
$u$ fails to satisfy PMP.

In its general form, the algorithm  \cite{Hale_2016} is defined in the framework of relaxed controls
(probability distributions on the space of ordinary controls) \cite{McShane_1967}.
However, in the setting of the problems formulated in this paper, it need only  compute ordinary controls.

\begin{algorithm}
\caption{Hamiltonian-based hill-climbing algorithm}\label{alg:algorithm}
\begin{algorithmic}[1]
\Procedure{}{} 
\State $k \gets 0$
\State $u_k \gets \text{Initial guess } u_0$
\While{$k < \texttt{iters}$}
\State Choose a finite grid $\mcal{G}_k \subset [0,T_f]$.
\State Solve forward for $\bm{y}_k(t)$ using $u_k$ 
\State Solve backwards for $\bm{\lambda}_k(t)$ using $u_k(t)$ and $\bm{y}_k(t)$
\State Compute, for every $t\in{\mcal G}_{k}$,  $u^*(t) = \argmax{u} H(\bm{y}_k,\bm{\lambda}_k,u)$; interpolate the results to define $u^*(t)$ $\forall$ $t\in[0,T_f]$.
\State Compute smallest $\ell =0,1,2\ldots$ s.t. \begin{equation} J(u_k) - J(u_k + \beta^\ell(u^* - u_k)) \leq \alpha\beta^\ell \Theta(u_k) \nonumber\end{equation}
\State $u_{k+1}(t) \gets u_k + \beta^\ell(u^* - u_k) $
\State $k \gets k+1$
\EndWhile\label{euclidendwhile}
\EndProcedure
\end{algorithmic}
\end{algorithm}

\section*{Acknowledgements}

This work is supported by ARO grant \#W911NF-14-1-0402 (to J.S.W).






\bibliographystyle{IEEEtran}
\bibliography{sources}

\begin{thebibliography}{10}
\providecommand{\url}[1]{#1}
\csname url@samestyle\endcsname
\providecommand{\newblock}{\relax}
\providecommand{\bibinfo}[2]{#2}
\providecommand{\BIBentrySTDinterwordspacing}{\spaceskip=0pt\relax}
\providecommand{\BIBentryALTinterwordstretchfactor}{4}
\providecommand{\BIBentryALTinterwordspacing}{\spaceskip=\fontdimen2\font plus
\BIBentryALTinterwordstretchfactor\fontdimen3\font minus
  \fontdimen4\font\relax}
\providecommand{\BIBforeignlanguage}[2]{{%
\expandafter\ifx\csname l@#1\endcsname\relax
\typeout{** WARNING: IEEEtran.bst: No hyphenation pattern has been}%
\typeout{** loaded for the language `#1'. Using the pattern for}%
\typeout{** the default language instead.}%
\else
\language=\csname l@#1\endcsname
\fi
#2}}
\providecommand{\BIBdecl}{\relax}
\BIBdecl

\bibitem{Dawes_1975}
R.~M. Dawes, ``Formal models of dilemmas in social decision-making,'' in
  \emph{Human Judgment and Decision Processes}.\hskip 1em plus 0.5em minus
  0.4em\relax Academic Press, 1975.

\bibitem{Ostrom_1990}
E.~Ostrom, \emph{Governing the Commons: The evolution of institutions for
  collective action}.\hskip 1em plus 0.5em minus 0.4em\relax Cambridge
  University Press, 1990.

\bibitem{Weitz_2016}
J.~S. Weitz, C.~Eksin, K.~Paarporn, S.~P. Brown, and W.~C. Ratcliff, ``An
  oscillating tragedy of the commons in replicator dynamics with
  game-environment feedback,'' \emph{Proceedings of the National Academy of
  Sciences}, vol. 113, no.~47, pp. E7518--E7525, 2016.

\bibitem{Brown_2007}
S.~P. Brown and F.~Taddei, ``The durability of public goods changes the
  dynamics and nature of social dilemmas,'' \emph{PLOS ONE}, vol.~2, no.~7, pp.
  1--7, 07 2007.

\bibitem{Roopnarine_2013}
P.~Roopnarine, ``Ecology and the tragedy of the commons,''
  \emph{Sustainability}, vol.~5, no.~2, pp. 749--773, 2013.

\bibitem{Hota_2016}
A.~R. Hota, S.~Garg, and S.~Sundaram, ``Fragility of the commons under
  prospect-theoretic risk attitudes,'' \emph{Games and Economic Behavior},
  vol.~98, pp. 135 -- 164, 2016.

\bibitem{Rapoport_1992}
A.~Rapoport and R.~Suleiman, ``Equilibrium solutions for resource dilemmas,''
  \emph{Group Decision and Negotiation}, vol.~1, no.~3, pp. 269--294, Nov 1992.

\bibitem{Hardin_1968}
\BIBentryALTinterwordspacing
G.~Hardin, ``The tragedy of the commons,'' \emph{Science}, vol. 162, no. 3859,
  pp. 1243--1248, 1968. [Online]. Available:
  \url{http://science.sciencemag.org/content/162/3859/1243}
\BIBentrySTDinterwordspacing

\bibitem{Aronson_2007}
J.~Aronson, S.~Milton, and J.~Blignaut, \emph{Restoring Natural Capital:
  Science, Business, and Practice}.\hskip 1em plus 0.5em minus 0.4em\relax
  Island Press, 2007.

\bibitem{Hota_2016_CDC}
A.~R. Hota and S.~Sundaram, ``Controlling human utilization of shared resources
  via taxes,'' in \emph{2016 IEEE 55th Conference on Decision and Control
  (CDC)}, Dec 2016, pp. 6984--6989.

\bibitem{Hutchings_2004}
J.~A. Hutchings and J.~D. Reynolds, ``Marine fish population collapses:
  Consequences for recovery and extinction risk,'' \emph{BioScience}, vol.~54,
  no.~4, pp. 297--309, 2004.

\bibitem{Penn_2003}
D.~Penn, ``The evolutionary roots of our environmental problems: Toward a
  darwinian ecology,'' \emph{The Quarterly Review of Biology}, vol.~78, no.~3,
  pp. 275--301, 2003.

\bibitem{VanVugt_2009}
M.~V. Vugt, ``Averting the tragedy of the commons: Using social psychological
  science to protect the environment,'' \emph{Current Directions in
  Psychological Science}, vol.~18, no.~3, pp. 169--173, 2009.

\bibitem{VanVugt_1999}
M.~van Vugt and C.~D. Samuelson, ``The impact of personal metering in the
  management of a natural resource crisis: A social dilemma analysis,''
  \emph{Personality and Social Psychology Bulletin}, vol.~25, no.~6, pp.
  735--750, 1999.

\bibitem{Lupia_2013}
A.~Lupia, ``Communicating science in politicized environments,''
  \emph{Proceedings of the National Academy of Sciences}, vol. 110, no.
  Supplement 3, pp. 14\,048--14\,054, 2013.

\bibitem{Stern_2000}
P.~C. Stern, ``New environmental theories: Toward a coherent theory of
  environmentally significant behavior,'' \emph{Journal of Social Issues},
  vol.~56, no.~3, pp. 407--424, 2000.

\bibitem{Manzoor_2014}
T.~Manzoor, S.~Aseev, E.~Rovenskaya, and A.~Muhammad, ``Optimal control for
  sustainable consumption of natural resources,'' \emph{IFAC Proceedings
  Volumes}, vol.~47, no.~3, pp. 10\,725 -- 10\,730, 2014, 19th IFAC World
  Congress.

\bibitem{Basar_1999}
T.~Basar and G.~J. Olsder, \emph{Dynamic noncooperative game theory}.\hskip 1em
  plus 0.5em minus 0.4em\relax Siam, 1999.

\bibitem{Hale_2016}
\BIBentryALTinterwordspacing
M.~Hale, Y.~Wardi, H.~Jaleel, and M.~Egerstedt, ``Hamiltonian-based algorithm
  for optimal control,'' \emph{ArXiv}, March 2016. [Online]. Available:
  \url{http://arxiv.org/abs/1603.02747}
\BIBentrySTDinterwordspacing

\bibitem{Liberzon}
D.~Liberzon, \emph{Calculus of Variations and Optimal Control Theory: A Concise
  Introduction}.\hskip 1em plus 0.5em minus 0.4em\relax Princeton University
  Press, 2012.

\bibitem{Armijo_1966}
L.~Armijo, ``Minimization of functions having lipschitz continuous first
  partial derivatives.'' \emph{Pacific J. Math.}, vol.~16, no.~1, pp. 1--3,
  1966.

\bibitem{Polak_1997}
E.~Polak, \emph{Optimization Algorithms and Consistent Approximations}.\hskip
  1em plus 0.5em minus 0.4em\relax New York, New York: Springer-Verlag, 1997.

\bibitem{McShane_1967}
E.~J. McShane, ``Relaxed controls and variational problems,'' \emph{SIAM
  Journal on Control}, vol.~5, no.~3, pp. 438--485, 1967.

\end{thebibliography}

\end{document}